%% file: main.tex
\documentclass[sigconf,review=false,nonacm]{acmart} 

\input{defs}

\input{solidity-highlighting}

\title{Transferable Cross-Chain Options}
\author{Daniel Engel}
\affiliation{%
   \institution{Computer Science Dept., Brown University}
   \city{Providence}
   \state{RI}
   \country{USA}}
\author{Yingjie Xue}
\affiliation{%
   \institution{Computer Science Dept., Brown University}
   \city{Providence}
   \state{RI}
   \country{USA}}

\begin{abstract}
\input{abstract}

\end{abstract}

\begin{document}
\maketitle

\section{Introduction}
\input{intro}

\section{Model}
\label{sec:model}
\input{model}

\section{Problem Overview}
\label{sec:overview}
\input{overview}

\section{Two-Party Transferable Swap}
\label{sec:two_party_swap_transfer}
\input{two_party_transfer_swap}

\section{Security Properties and Proof}
\seclabel{sec:properties}
\input{properties}

\section{Related Work}
\seclabel{related}
\input{related}

\section{Remarks and Conclusions}
\label{sec:conclusions}
\input{conclusions}


\bibliographystyle{ACM-Reference-Format}
\bibliography{references,bibliography}
\clearpage
\newpage
\appendix
\section*{Appendix}

\input{code}

\input{proof}

\end{document}

%% file: defs.tex
\usepackage{listings}
\usepackage{todonotes}
\usepackage{subcaption}
\usepackage{thm-restate}
\newcommand{\var}[1]{\text{\lstinline+#1+}}
\newtheorem{theorem}{Theorem}

\usepackage{epsfig}

\newcommand{\Carol}{\ensuremath{\text{Carol}}}
\newcommand{\David}{\ensuremath{\text{David}}}

\newcommand{\seclabel}[1]{\label{sec:#1}}
\newcommand{\nakedsecref}[1]{\ref{sec:#1}}
\newcommand{\secref}[1]{Section~\nakedsecref{#1}}

\usepackage{algorithm}
\usepackage{algorithmic}
\usepackage{xltabular} 
\graphicspath{{./pictures/}}

\newcommand{\addressVar}{\ensuremath{\mathit{address}}}
\newcommand{\approveLeader}{\ensuremath{\mathit{approveLeader}}}
\newcommand{\Asset}{\ensuremath{\mathit{Asset}}}
\newcommand{\candidateReceiver}{\ensuremath{\mathit{candidate\_receiver}}}
\newcommand{\candidateSender}{\ensuremath{\mathit{candidate\_sender}}}
\newcommand{\claim}{\ensuremath{\mathit{claim}}}
\newcommand{\contestLeader}{\ensuremath{\mathit{contestLeader}}}
\newcommand{\contest}{\ensuremath{\mathit{contest}}}
\newcommand{\mutateLockFollower}{\ensuremath{\mathit{mutateLockFollower}}}
\newcommand{\mutateLockFreeFollower}{\ensuremath{\mathit{mutateLockFreeFollower}}}
\newcommand{\mutateLockLeader}{\ensuremath{\mathit{mutateLockLeader}}}
\newcommand{\mutation}{\ensuremath{\mathit{mutation}}}
\newcommand{\receiver}{\ensuremath{\mathit{receiver}}}
\newcommand{\refund}{\ensuremath{\mathit{refund}}}
\newcommand{\replaceFollower}{\ensuremath{\mathit{replaceFollower}}}
\newcommand{\replaceHashlock}{\ensuremath{\mathit{replace\_hashlock}}}
\newcommand{\replaceLeader}{\ensuremath{\mathit{replaceLeader}}}
\newcommand{\replace}{\ensuremath{\mathit{replace}}}
\newcommand{\revertLeader}{\ensuremath{\mathit{revertLeader}}}
\newcommand{\sender}{\ensuremath{\mathit{sender}}}
\newcommand{\startFollower}{\ensuremath{\mathit{startFollower}}}
\newcommand{\startLeader}{\ensuremath{\mathit{startLeader}}}
\newcommand{\startTime}{\ensuremath{\mathit{start\_time}}}
\newcommand{\swapHashlock}{\ensuremath{\mathit{swap\_hashlock}}}
\newcommand{\timeoutBob}{\ensuremath{\mathit{timeout\_Bob}}}
\newcommand{\timeoutCarol}{\ensuremath{\mathit{timeout\_Carol}}}
\newcommand{\timeout}{\ensuremath{\mathit{timeout}}}

%% file: solidity-highlighting.tex
\usepackage{listings, xcolor}

\definecolor{verylightgray}{rgb}{.97,.97,.97}

\lstdefinelanguage{Solidity}{
	keywords=[1]{anonymous, assembly, assert, balance, break, call, callcode, case, catch, class, constant, continue, constructor, contract, debugger, default, delegatecall, delete, do, else, emit, event, experimental, export, external, false, finally, for, function, gas, if, implements, import, in, indexed, instanceof, interface, internal, is, length, library, log0, log1, log2, log3, log4, memory, modifier, new, payable, pragma, private, protected, public, pure, push, require, return, returns, revert, selfdestruct, send, solidity, storage, struct, suicide, super, switch, then, this, throw, transfer, true, try, typeof, using, value, view, while, with, addmod, ecrecover, keccak256, mulmod, ripemd160, sha256, sha3}, 
	keywordstyle=[1]\color{blue}\bfseries,
	keywords=[2]{address, bool, byte, bytes, bytes1, bytes2, bytes3, bytes4, bytes5, bytes6, bytes7, bytes8, bytes9, bytes10, bytes11, bytes12, bytes13, bytes14, bytes15, bytes16, bytes17, bytes18, bytes19, bytes20, bytes21, bytes22, bytes23, bytes24, bytes25, bytes26, bytes27, bytes28, bytes29, bytes30, bytes31, bytes32, enum, int, int8, int16, int24, int32, int40, int48, int56, int64, int72, int80, int88, int96, int104, int112, int120, int128, int136, int144, int152, int160, int168, int176, int184, int192, int200, int208, int216, int224, int232, int240, int248, int256, mapping, string, uint, uint8, uint16, uint24, uint32, uint40, uint48, uint56, uint64, uint72, uint80, uint88, uint96, uint104, uint112, uint120, uint128, uint136, uint144, uint152, uint160, uint168, uint176, uint184, uint192, uint200, uint208, uint216, uint224, uint232, uint240, uint248, uint256, var, void, ether, finney, szabo, wei, days, hours, minutes, seconds, weeks, years},	
	keywordstyle=[2]\color{teal}\bfseries,
	keywords=[3]{block, blockhash, coinbase, difficulty, gaslimit, number, timestamp, msg, data, gas, sender, sig, value, now, tx, gasprice, origin},	
	keywordstyle=[3]\color{violet}\bfseries,
	identifierstyle=\color{black},
	sensitive=false,
	comment=[l]{//},
	morecomment=[s]{/*}{*/},
	commentstyle=\color{gray}\ttfamily,
	stringstyle=\color{red}\ttfamily,
	morestring=[b]',
	morestring=[b]"
}

%% file: abstract.tex
An \emph{option} is a financial agreement between two parties to trade two assets.
One party is given the right, but not the obligation,
to complete the swap before a specified termination time.
In today's financial markets,
an option is considered an asset which can itself be transferred:
while an option is active,
one party can sell its rights (or obligations) to another.

Today's blockchains support simple options in the form of
cross-chain atomic swap protocols where one party has the
choice whether to complete the swap.
The options implemented by these cross-chain protocols,
are not, however, transferable.

This paper proposes novel distributed protocols for
\emph{transferable cross-chain options},
where both option owners and providers can sell their positions to third parties.
The protocol ensures that none of the parties can be cheated,
that no unauthorized party can interfere,
and that the transfer succeeds if the buyer and seller faithfully follow the protocol.

%% file: intro.tex
An \emph{option} is a financial agreement between two parties,
say, Alice and Bob.
Alice owns some units of ``florin'' cryptocurrency,
while Bob owns some units of ``guilder'' cryptocurrency.
Alice and Bob agree that Alice will purchase from Bob
the right, but not the obligation,
to buy 100 guilders from him in return for 100 florins from her,
at any time before, say, next Tuesday.
If the value of guilders relative to florins goes up before Tuesday,
Alice will exercise her option by executing the trade,
and otherwise she will keep her florins and walk away.
Alice pays Bob a fee, called a premium,
to compensate him for his inconvenience and risk.

In the world of \emph{decentralized finance} (Defi),
florins and guilders are managed on distinct blockchains,
Alice and Bob are autonomous agents, 
and they do not trust each other.
Moreover, 
they have no recourse to third-party arbiters or to courts of law.
Several non-trivial distributed protocols have been proposed to execute
cross-chain swaps and options~\cite{Herlihy2018,HerlihyLS2019,XueH2021,arwen2019,eizinger,han2019optionality,liu2018,tefaghcapital,xu2020game}.
Simplifying somewhat,
these protocols guarantee that if Alice and Bob are both honest,
then the deal unfolds as planned,
but if either party cheats,
the honest party ends up in what it considers a satisfactory position.

The contribution of this paper is to take the notion of a cross-chain option
to the next level:
while Alice's unexercised option has value,
she should be able to sell it to a third party, Carol.
In conventional finance, such a transfer is simple:
Bob grants Alice an options contract,
Alice signs over that contract to Carol,
and all agreements are enforced by civil law.
In the lawless world of decentralized finance,
by contrast,
transferable options require a carefully-crafted distributed protocol.

We propose novel cross-chain protocols,
to support transferable cross-chain options.
Additionally, we prove the security properties of our proposed protocols, and detailed pseudocode for the protocols is provided.
Informally, our proposed protocols can address option transfer in the following scenarios:
\begin{enumerate}
    \item Alice (the option owner, i.e the one who has a right to buy an asset) transfers her position to another party (Carol).
    \item Bob (the option provider, i.e. the one who provides the owner the right) transfers his position to another party (David).
    \item Alice and Bob concurrently transfer their positions to Carol and David, respectively.
\end{enumerate}

Take the option position transfer between Alice and Carol for example. The proposed transfer protocol ensures that if Alice and Carol are honest,
then Alice relinquishes her rights,
Carol assumes Alice's rights,
and Bob cannot veto the transfer.
If Carol cheats,
then Alice's rights remain with Alice.

We view the proposed distributed protocols as a first step toward a more ambitious goal.
One can imagine more complex cross-chain deals where
parties acquire various rights and obligations
(options, futures, derivatives, and so on).
It should be possible for a party who holds unrealized rights or obligations
to sell those rights or obligations to another party,
atomically transferring its position
without being hindered by any third party.
Here, 
we show how to make certain transferable cross-chain option deals,
but someday we hope to do the same for arbitrary cross-chain deals.

This paper is organized as follows. Our model is given in \secref{model}. We provide an overview of our proposed transferable two-party swap protocols in \secref{overview} and detailed protocols are described in \secref{two_party_swap_transfer}. We prove security properties of our proposed protocols in \secref{sec:properties}.  In \secref{related}, we describe related work. Finally, we conclude in  \secref{conclusions}.

%% file: model.tex
For our purposes,
a \emph{blockchain} is a tamper-proof distributed ledger (or database) that
tracks ownership of \emph{assets} by \emph{parties}.
An asset can be a cryptocurrency, a token, an electronic deed to property, and so on.
A party can be a person, an organization, or even a contract (see below).
There are multiple blockchains managing different kinds of assets.
We focus here on applications where mutually-untrusting parties trade assets
across multiple blockchains.
Because we treat blockchains simply as ledgers,
our results do not depend on specific blockchain technology
(such as proof-of-stake vs proof-of-work).
We assume only that ledgers are highly available, tamper-proof,
and capable of running smart contracts.

A \emph{smart contract} (or ``contract'') is a
blockchain-resident program initialized and called by the parties.
A party can create a contract on a blockchain,
or call a function exported by an existing contract.
Contract code and state are public,
so a party calling a contract function knows what code will be executed.
Contract code is deterministic because contract execution is replicated,
and all executions must agree.

A contract can read or write ledger entries on the blockchain where it resides,
but it cannot actively access data from the outside world,
including calling contracts on other blockchains.
Although there are protocols that allow blockchains to communicate (cross-chain proofs), they have weaknesses that make them difficult to use in practice: for example, incompatibility problems, lack of decentralization, and non-deterministic guarantees of message delivery ~\cite{robinson2021survey}. Therefore, we assume different blockchains do not communicate.

A contract on blockchain $A$ can learn of a change to
a blockchain $B$ only if some party explicitly informs $A$ of $B$'s change, 
along with some kind of ``proof'' that the information about $B$'s state is correct.
Contract code is passive, public, deterministic, and trusted,
while parties are active, autonomous, and potentially dishonest.

At any given time each party has a given financial \textit{position}.
This is characterized by their current and future financial obligations/allowances.
Financial positions are managed by rules on contracts that can reside across multiple distinct blockchains.

Our execution model is \emph{synchronous}:
there is a known upper bound $\Delta$ on the propagation time
for one party's change to the blockchain state,
plus the time to be noticed by the other parties.
Specifically, blockchains generate new blocks at a steady rate,
and valid transactions sent to the blockchain will be included in a block
and visible to participants within a known, bounded time $\Delta$.

As noted,
we assume blockchains are always available,
tamper-proof,
and that they correctly execute their contracts.
Although parties may display Byzantine behavior,
contracts can limit their behavior by rejecting unexpected contract calls.

We make standard cryptographic assumptions.
Each party has a public key and a private key,
and any party's public key is known to all.
Messages are signed so they cannot be forged,
and they include single-use labels (``nonces'')
so they cannot be replayed.

%% file: overview.tex
One way to convey the challenge presented by making options transferable
is to trace the evolutionary path that leads from 20th Century distributed
systems to modern decentralized finance (DeFi).
Suppose, in 1999, Alice and Bob agree that she will transfer 100 dollars to him
at a New York bank,
and he will transfer 100 euros to her at a Paris bank.
Each bank maintains its own database, which communicate over a network.
The swap protocol must tolerate hardware failures:
databases can crash and messages can be lost.
Each bank makes a tentative transfer.
In the classical \emph{two-phase} commit protocol~\cite{BernsteinHZ1986},
each bank records its tentative transfers ``somewhere safe''
(on a magnetic disk that survives crashes)
and sends to a trusted coordinator a vote whether to commit or abort.
If both banks vote to commit, the transfers are installed,
and if either votes to abort, or does not vote in a reasonable time,
the transfers are discarded, leaving both databases unchanged.

The two-phase commit protocol established a pattern for later blockchain protocols.
Today, Alice and Bob want to trade units of cryptocurrency.
They agree to exchange 100 of her (electronic) florins for 100
of his (electronic) guilders.
Each cryptocurrency is managed on a distinct blockchain.
Alice and Bob must agree on a swap protocol that
tolerates not only hardware failures,
but Byzantine failures by participants:
each party must protect itself if its counterparty cheats
by departing from the agreed-upon protocol.

Atomic swap protocols based on \emph{hashed timelock contracts} (HTLCs)~\cite{tiernolan,Herlihy2018} mimic two-phase commit.
In the HTLC protocol's first phase,
each party places the assets to be transferred ``somewhere safe''.
In place of writing to magnetic disk,
Alice transfers her florins to an \emph{escrow} account
controlled by a \emph{smart contract},
a program that decides, based on later input,
whether to \emph{commit} by transferring Alice's florins to Bob,
or to \emph{abort} by refunding Alice her florins.
Then Bob does the same.
Each escrow transfer is effectively a vote to commit.
When both votes are confirmed,
Alice releases a secret that causes the contracts to complete the swap.
Otherwise, if the escrow transfers are not confirmed in reasonable time,
the contracts refund the assets to their original owners.
This description elides many technical details,
but the HTLC protocol's overall structure is remarkably similar to
two-phase commit, despite substantial differences in their failure models.

While technically correct, HTLC swap protocols are flawed:
once both assets are escrowed,
Alice can take her time deciding whether to trigger the swap.
Cryptocurrencies are notoriously volatile,
so if the value of Bob's guilders goes up relative to Alice's florins
before the timeout expires,
she can choose to complete the swap at the last minute.
If the value goes down, she is free to walk away without penalty.
Bob may be reluctant to accept such a deal.

This protocol is unfair to Bob because only Alice has \emph{optionality}:
at the end, he cannot back out of the deal, but she can.
If she does back out, he gets his guilders back,
but only after a possibly long delay while the market is moving against him.
Bob thus incurs the opportunity cost of not being able to use his coins while they are escrowed.
By contrast, in conventional finance,
this deal would be structured as an \emph{option contract},
where Alice pays Bob a fee, called a \emph{premium},
to compensate him if she walks away without completing the deal.

Cross-chain atomic swap options require carefully-defined
distributed protocols~\cite{liu2018,XueH2021}.
The protocol proposed by Xue and Herlihy~\cite{XueH2021}
is structured like an iterated two-phase commit protocol:
in the first phase the parties escrow premiums.
If all goes well, in the next phase they escrow coins.
If all goes well in the final phase they complete the swap.
Any party who drops out of the protocol ends up
paying a premium to the other,
and in the end, Alice has the optionality,
but Bob has been compensated for his risk.

We are now ready to address the main topic of this paper.
As long as Alice has optionality, that optionality has value.
Alice should be able to sell her position to a third party Carol.
(Bob may also want to sell his position to a third party David). For starters, we focus on Alice's position transfer since Bob's transfer will be similar.
There are many reasons Alice and Carol might agree to such a deal.
Alice may want to liquidate her position because she needs cash.
Perhaps Alice and Carol have different opinions on the future values of
florins versus guilders,
or they have different levels of risk tolerance.

A \emph{transferable atomic swap option} (or ``transferable swap'') protocol
roughly must satisfy the following properties.
(We define these properties more precisely in \secref{sec:properties}.)
\begin{itemize}
\item \emph{Liveness}: If Alice and Carol both follow the protocol,
  then (1) Carol acquires the right to buy Bob's coins at the same price and deadline,
  (2) Alice loses that right, and she is paid by Carol,
  (3) Bob cannot veto the transfer.
\item \emph{Safety}: As long as one of Alice or Carol follows the protocol,
  (1) the protocol completes before the option expires, and
  (2) when the protocol completes,
  exactly one of Alice or Carol has the right to trigger the atomic swap with Bob, (3) Bob's position in the option does not change if he conforms. 
\end{itemize}
Here is a high-level summary of our protocol.
The contract linking Alice and Bob is a delayed swap:
Alice escrows her asset, then Bob escrows his.
These escrow contracts are controlled by a \emph{hashlock} $h$.
Until each contract's timeout expires,
that contract will complete its side of the swap
when Alice produces a \emph{secret} $A_1$ such that $h = H(A_1)$,
where $H(\cdot)$ is a predefined cryptographic hash function.
Alice has optionality because she alone knows the secret.

Here is how Carol can buy Alice's position.
First,
there must be enough time to complete this transfer before
Alice's option expires.
Roughly speaking,
Carol generates a secret,
then Alice and Carol swap the roles of Carol's and Alice's secrets:
following the transfer, Alice's secret will no longer trigger the swap with Bob,
but Carol's secret will.

There are two contracts (discussed in detail below).
Contract $AB$ controls Alice's possible payment to Bob,
while $BA$ controls Bob's payment to Alice.
The optionality transfer itself is structured like a cross-chain swap,
exchanging roles instead of assets.
In the first phase,
Alice marks contracts $AB$ and $BA$ as \emph{mutating},
temporarily preventing any transfer to or from Bob.
This step prevents Alice or Carol from creating a chaotic
situation by triggering a partial asset swap with Bob while the
optionality transfer is in progress.
In the second phase,
Carol replaces the (hash of) Alice's secret and address
in both $AB$ and $BA$ with the (hash of) Carol's secret and address.
Just as for regular swaps,
the optionality transfer protocol will time out and revert if Alice or Carol
fails to take a step in time.

There is still a danger that Alice and Carol might cheat Bob
by making inconsistent changes to contracts $AB$ and $BA$.
Because $AB$ and $BA$ cannot coordinate directly,
the protocol has a built-in delay to give Bob an opportunity to \emph{contest}
a malformed transfer and to revert its changes.
If the transfer is well-formed,
Bob can expedite the protocol by actively approving of the transfer,
perhaps in return for an extra premium.
If Bob remains silent, the protocol will proceed after the delay expires.
Bob contests the transfer by providing proof that Alice signed inconsistent
changes to $AB$ and $BA$.

To keep the presentation uncluttered,
we omit some functionality that would be expected in a full protocol,
but that is not essential for optionality transfer.
For example,
there would be additional steps where parties deposit premiums
to compensate one another if one party leaves the other's assets temporarily trapped in escrow,
where Alice pays Bob a premium to encourage prompt transfer approval,
and where Alice posts a bond to be slashed if she is caught sending
inconsistent information to $AB$ and $BA$.

%% file: two_party_transfer_swap.tex
\label{protocol}
\subsection{Overview}

We will now describe several protocols for transferring option positions.
First we will describe a protocol that allows Alice to transfer her position as the option owner in a swap with Bob, to a third-party Carol.
Next, we will give a protocol that allows Bob to transfer his position as the option provider in a swap with Alice, to a third-party David.
Finally, we show that both of these protocols can be run concurrently and can also support multiple concurrent buyers.
Solidity code for these contracts appears in Appendix \ref{sec:code}.

A party that follows the protocol is \emph{conforming},
and a party that does not is \emph{adversarial}.
A party's \emph{principal} is the assets escrowed at a contract
during the protocol execution.
We use \emph{transfer} and \emph{replace} interchangeably.
First we describe the initial setup where Alice and Bob first lock their principals using a hashlock.
Throughout this paper, the party who owns the preimage to the hashlock (Alice,the initial owner of the option) is called a \textit{leader} and the counterparty (Bob) is called a \textit{follower}.
This terminology is similar to prior two-party, HTLC-based cross-chain swap
protocols~\cite{tiernolan,decred2,sparkswap}.

\paragraph{Setup}
\begin{enumerate}
    \item Agreement. Alice and Bob agree on a time $\startLeader$ to start the protocol, usually shortly after this agreement. Alice and Bob agree on the amounts  $\Asset_A$, $\Asset_B$ forming their principals,
as well as the number of rounds $dT$\footnote{We inherit the requirement $dT \geq 4$ from the standard two-party swap protocol.}, where after $T = \startLeader + dT \cdot \Delta$ when Alice's optionality on the $BA$ contract expires. 
    Alice generates a secret $A_1$ and $H(A_1)$ which is used as $\swapHashlock$ to redeem both principals as in a typical two-party swap.
    \item  Escrow. After the agreement achieved, 
    
    \begin{enumerate}
    \item Within $\Delta$, Alice publishes the
    $AB$ contract, escrowing $\Asset_A$, setting $H(A_1)$ as the $\swapHashlock$ and set $T+\Delta$ as the timeout for $AB$ contract to expire. 
    
    \item \sloppy Once Bob sees $\Asset_A$ is escrowed with the correct timeout, before $\Delta$ elapses, Bob creates the $BA$ contract, escrowing $\Asset_B$, setting $H(A_1)$ as the $\swapHashlock$ and set $T$ as the timeout for $BA$ contract to expire.
    \end{enumerate}

\end{enumerate}

A swap option between Alice and Bob, where Alice has optionality,
is determined by Alice and Bob's addresses, i.e. $\sender$ and $\receiver$ in contracts,
the $\swapHashlock$ used for redemption, and the timeout $T$ for refunds.
To make the swap option transferable to Carol,
we will need to replace Alice's address with Carol's as $\sender$ in the $AB$ contract 
and as $\receiver$ in the $BA$ contract, and replace the $\swapHashlock$ with a new one generated by Carol.
\sloppy In short, we must atomically replace the fields $(AB.\sender, BA.\receiver,AB.\swapHashlock, BA.\swapHashlock)$.

The contracts export the following functions of interest.
\begin{itemize}
\item Asset transfer related: the receiver calls $\claim()$ to withdraw the asset in the contract.
  This function requires that the preimage of the $\swapHashlock$ is sent to $\claim()$ before the contract expires.
  The sender calls $refund()$ to withdraw the asset after the contract expires. 
    \item Option transfer related: $\mutateLockLeader()$ temporarily freezes the assets and stores a $\replaceHashlock$ used to replace Alice's role (the replacement between Carol and Alice is also enforced by a hashlock mechanism) and the new $\swapHashlock$ if Carol successfully replaces Alice. 
The call to $\replace()$ completes the replacement for Carol.
\end{itemize}

For example, consider the $AB$ contract.
The \mutateLockLeader() function freezes the assets,
tentatively records the new $\sender$ as $\candidateSender$,
the new $\swapHashlock$ for redemption,
and records a short-term $\replaceHashlock$ to be used by $\replace()$
to finalize the replacement. Later, when $\replace()$ is called with the secret
matching the $\replaceHashlock$, these tentative changes become permanent.
The $\replaceLeader()$ function has a timeout.
If that timeout expires,
the internal $\revertLeader()$ function unfreezes the assets, discards the tentative changes, 
and restores the contract's previous state.

The challenge is how to ensure the transfer is atomic on both contracts.
If the transfer is not atomic,
Bob will be cheated if his principal can be claimed but he cannot claim Alice's principal.
To protect Bob,
the contract allows Bob to inspect the changes
and ensure atomicity by relaying changes on one contract to another.

\subsection{Transfer Leader Position}
\label{transfer_leader}

Our protocol will consist of 3 phases, each of which we describe in turn:
\begin{enumerate}
    \item \emph{Mutate Lock Phase}. Alice locks the assets on both contracts and tentatively transfers her position to Carol.
    \item \emph{Consistency Phase}. Bob makes sure the tentative changes on both contracts are consistent.
    \item \emph{Replace/Revert Phase}. Carol replaces Alice's role or she gives up and Alice gets her option back.
\end{enumerate}
\paragraph{I: Mutate Lock Phase}

Assume Alice and Carol agree on $\Asset_C$ for transferring her role to Carol at time $\startLeader$, which is the start time the following executions are based on\footnote{The protocol requires that $\startLeader \leq T-9\Delta$. Otherwise, there may not be sufficient time to complete the replacement.}. Carol has secrets $C_1,C_2$ and generates hashlocks $H(C_1),H(C_2)$. Let $C_{msg}= [H(C_1),H(C_2)]$ be the message Carol sends to specify that she would like to use $[\replaceHashlock$, $\swapHashlock]$, respectively.

In this phase, Carol prepares tentative payment $\Asset_C$ to Alice to replace Alice's role. Alice locks the $AB$ and $BA$ contracts so that the currently escrowed assets cannot be claimed from the original swap.

\begin{enumerate}
    \item  By $\startLeader+\Delta$, Carol creates the $CA$ contract, escrowing $\Asset_C$, setting the $\swapHashlock$ as $H(C_1)$ for Alice to redeem  $\Asset_C$, and the timeout to be $\startLeader+9\Delta$.  
    She also sends $C_{msg} = [H(C_1),H(C_2)]$ to Alice. 
    
    \item If the previous step succeeds, before $\Delta$ elapses,
    Alice concurrently calls $\mutateLockLeader()$ on both $AB$ and $BA$ contracts, tentatively setting $[\replaceHashlock=H(C_1),\swapHashlock=H(C_2)]$ to both contracts, and  setting $\candidateSender = \Carol.\addressVar$ on $AB$ and $\candidateReceiver = \Carol.\addressVar$ on $BA$ contract. If Alice is conforming, $\mutateLockLeader()$ should be called by $\startLeader+2\Delta$ on both $AB$ and $BA$.
\end{enumerate}

\paragraph{II: Consistency Phase}

It is possible that in the worst case, in the Mutate Lock Phase, Alice could report inconsistent tentative changes to $AB,BA$ by either mutating one contract but not the other, or by reporting different hashlocks to each contract.
Because the two contracts cannot communicate, they have no way of knowing what is happening on the other contract.
One such example of an attack is shown in  Fig. \ref{figure:execution}.

Because of these attacks, in the \emph{Consistency Phase}, Bob is given a period to ensure both $AB$ and $BA$ have the same changes.
If one of the contracts is not changed, then Bob forwards the change (by forwarding Alice's signature on the contract she signed) to the other contract. 
If both contracts are tentatively changed but the changed are different, then Bob can call $\contestLeader()$ to prove that Alice has lied and reported inconsistent changes signed by Alice.

\begin{enumerate}
    \item If Bob sees that Alice only calls $\mutateLockLeader()$ on $AB(BA)$, then within $\Delta$, Bob should call 
    $\mutateLockLeader()$ on $BA(AB)$ contract and set $[\replaceHashlock=H(C_1),\swapHashlock=H(C_2)]$ by forwarding Alice's signature.

    \item If Bob sees Alice has called $\mutateLockLeader()$ on both $AB$ and $BA$ , but the changes are different, i.e. either one or more hashlocks are not the same, or the new candidate sender does not correspond to the new candidate receiver, then before a $\Delta$ elapses, Bob should call $\contestLeader()$ on both contracts, forwarding Alice's signature on one contract to the other. In this case, the tentative change will be reverted.

    \item If Bob sees Alice has called $\claim()$ with her secret $A_1$ on $BA$ to redeem $\Asset_B$, and she also successfully called $\mutateLockLeader()$ on $AB$, then Bob should call $\contestLeader()$ on $AB$  before a $\Delta$ elapses to revert the tentative change.
    Bob then calls $\claim()$ using $A_1$ on $AB$ to redeem $\Asset_A$.
\end{enumerate}

Importantly, Bob is able to contest when either Alice reports inconsistent mutation signatures between contracts or if she preemptively reveals her swap secret.
Bob's ability to call $\contestLeader()$ ensures that
he's able to maintain his position in the original trade if Alice deviates from the protocol by reporting inconsistent changes. 

\paragraph{III: Replace/Revert Phase}

After the Consistency Phase has ended, in the Replace/Revert Phase, the swap either finalizes to one between Carol and Bob or reverts back to the original swap between Alice and Bob.
Once Carol sees consistent changes with no successful $\contestLeader()$ calls by Bob, she calls $\replaceLeader()$ to make herself the new leader of the swap. Figure~\ref{fig:all_conforming} shows an execution of the protocol where all parties are conforming.
If Bob sees Carol only call $\replaceLeader()$ on one contract, he uses her revealed secret $C_1$ to finalize the replacement on the other contract.

\begin{enumerate}
    \item When Carol sees the tentative changes are the same on both $AB$ and $BA$ contract, and Bob is not able to contest anymore on both contracts \footnote{If Bob mutates a contract first, there is no contest window on that contract (He is implicitly approving of the tentative changes). If Alice mutates a contract first, Bob is given $2\Delta$ to dispute that mutation.}, then before $\Delta$ elapses, she calls $\replaceLeader()$ on both contracts with her secret $C_1$.
    \item If Bob sees Carol only called $\replaceLeader()$ on one contract $BA(AB)$ , then within $\Delta$ he calls $\replaceLeader()$ on $AB(BA)$.
    \item If Carol gives up the replacement and does not call $\replaceLeader()$ when she can, i.e. $(now - AB(BA).\mutation.\startTime > 6*\Delta)$, then the tentative changes can be reverted by $\revertLeader()$ causing the assets to be unfrozen.
    \item If Alice sees $C_1$ which is passed by $\replaceLeader()$, then within $\Delta$ she calls $\claim()$ to claim $\Asset_C$ on the $CA$ contract.
\end{enumerate}

The Replace/Revert Phase marks the point of the protocol where Alice's position is actually transferred to Carol.  
\begin{figure}
    \centering
    \includegraphics[scale =1, width=\linewidth]{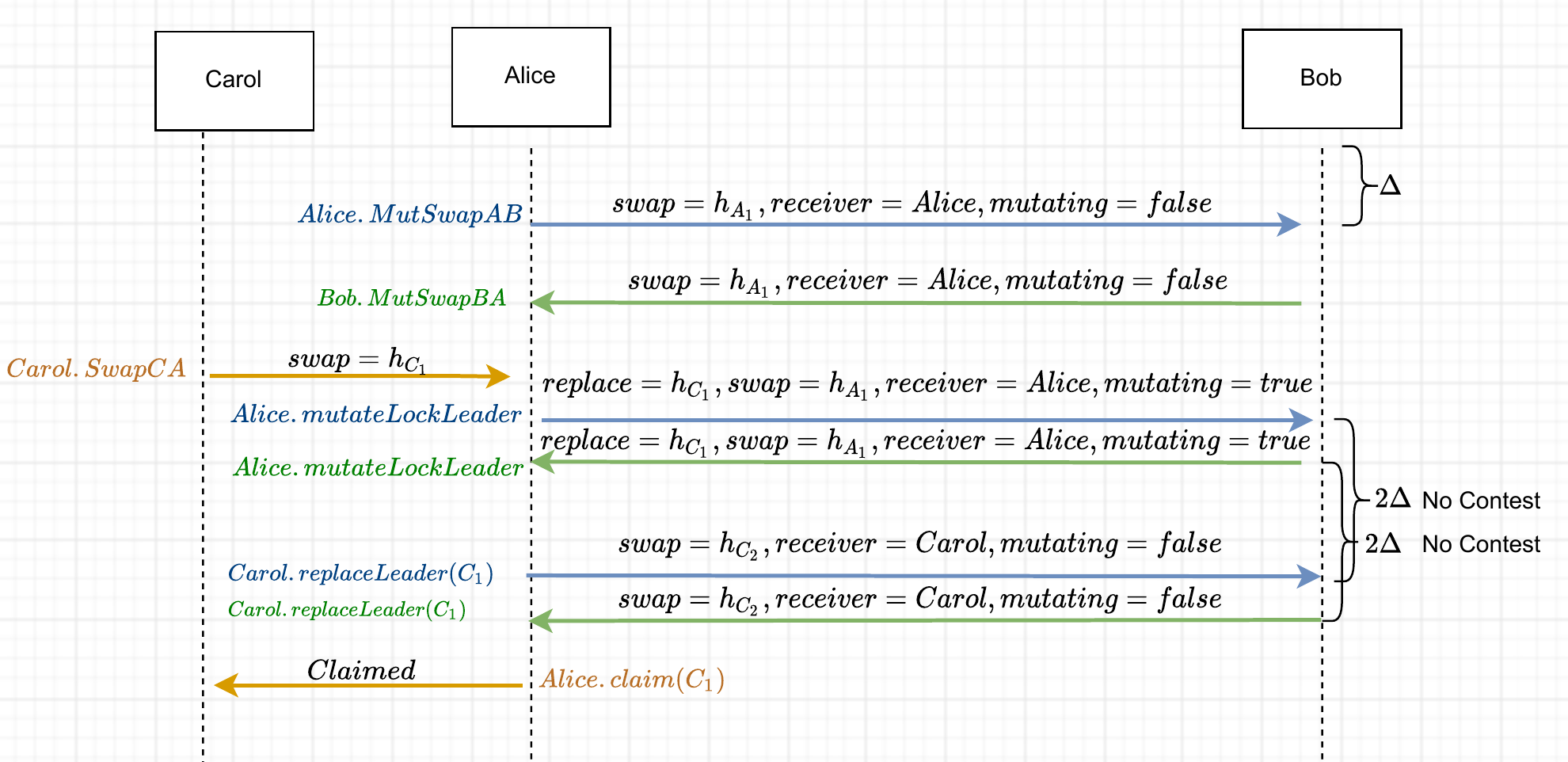}
    \caption{In the figure, $x.Function$ means party $x$ calls $Function$ in the contract. E.g. $Alice.MutSwapAB$ means Alice creates the swap contract $AB$ and escrows her assets. The blue arrow depicts contract $AB$ and the green arrow depicts  contract $BA$ and the orange one depicts contract $CA$. The text above each arrow depicts the state of the contract. For example, $swap =h_{A_1}, receiver= Alice, mutating =false$ means $swap\_hashlock =h_{A_1}$, and Alice can claim the asset in the contract if she provides preimage of $h_{A_1}$. $replace$ is short for $replace\_hashlock$.  }
    \label{fig:all_conforming}
\end{figure}
Before this phase, Alice has only set up a tentative transfer to Carol. If Carol changes her mind and gives up the replacement, the tentative swap between Carol and Bob can revert back to the original one between Alice and Bob by calling $\revertLeader()$. Instead of calling $\revertLeader()$ directly, Alice and Bob can alternatively call $\mutateLockLeader()$,$\refund()$,or $\claim()$ at this point in the protocol since these will all automatically call $\revertLeader()$.
Alice  would do this if after an unsuccessful attempt to transfer her swap, she wants to attempt another transfer, reclaim her escrowed funds, or exercise the swap.

\paragraph{Timeouts}

Timeouts are critical to guarantee the correctness of our protocol. They were omitted in the earlier protocol descriptions for simplicity. Here, we provide the timeouts we set in each step.

The $BA$ contract expires at time $T$,
the last time for the leader to send her secret to redeem the principal.
The $AB$ contract expires at time $T+\Delta$.

Denote the time when $\mutateLockLeader()$ is called on $AB(BA)$ contract  as $AB(BA).\mutation.\startTime$.
After the mutation starts on any contract, Bob is given time $2\Delta$ to contest Alice's mutation:

$AB(BA).\contest.\timeout=AB(BA).\mutation.\startTime+2\Delta$.

When Bob is able to contest, and Bob has not called $\mutateLockLeader()$ himself, Carol cannot call $\replaceLeader()$. After the contest period elapses, Carol has $2\Delta$ to call $\replaceLeader()$. Here we have $AB(BA).\replace.\timeoutCarol = AB(BA).\mutation.\startTime+4\Delta$.

If Carol deviates and only calls $\replaceLeader()$ on one contract, we allow Bob to call $\replaceLeader()$ on the another contract. We give Bob $2\Delta$ more than Carol to call $\replaceLeader()$.   That is, $AB(BA).\replace.\timeoutBob = AB(BA).\mutation.\startTime+6\Delta$.

We see that it takes $6\Delta$ from the start of a $\mutateLockLeader()$ call to finalize a replacement in the worst case. Thus, $AB(BA).\mutation.\startTime+6\Delta \leq T-\Delta$, which means $AB(BA).\mutation.\startTime \leq T-7\Delta $.

The last person that is able to call $\mutateLockLeader()$ is Bob.
He is given one more $\Delta$ than Alice  on each contract to call this function in case Alice decides to just call it on one contract.
The deadline for Alice to call $\mutateLockLeader()$ should be $T - 7\Delta$.
The deadline for Bob to call $\mutateLockLeader()$ should be $T - 6\Delta$.
Consider the case when Alice is adversarial by not calling $\mutateLockLeader()$ on $BA$ but does call it on $AB$ by $T - 7\Delta$.
If Bob is compliant, then he calls $\mutateLockLeader()$  by  $T-6\Delta$ on the $BA$ contract.
Because we require $AB.\mutation.\startTime \leq T-7\Delta$, Bob is guaranteed $6\Delta$ to a do a full replacement in the worst case.
Conforming Carol's safety is also guaranteed since she can give up the replacement if she finds there is not sufficient time to call $\replaceLeader()$ then $\claim()$. In that case, due to the timeout of the original swap, adversarial Alice loses her opportunity to sell her position.

For CA edge, after Carol escrows $\Asset_C$, in the worst case it takes $7\Delta$ to complete the replacement, ($\Delta$ to start the mutation and $6\Delta$ to complete the replacement). After the replacement, it takes one $\Delta$ for Alice to redeem. Thus, the timeout for $CA$ edge is $\startLeader+9\Delta$.

The reason why we have $2\Delta$ for Bob to contest, and $2\Delta$ more than Carol to call $\replaceLeader()$ is that the start time of mutation can be staggered by $\Delta$ on two contracts, which will be described in detail in proof.

\subsubsection{An Optimization}
\label{transfer_leader_altruistic}

Because of the $\contestLeader()$ time window, in the worst case, Carol has to wait $2 \Delta$ after a $\mutateLockLeader()$ call to decide whether to release her secret or not. We can speed up the process by adding an $\approveLeader()$ method. 
Instead of waiting for $2\Delta$ even though Bob does not contest at all, once Bob sees Alice has made consistent mutations on both edges, Bob can call $\approveLeader()$ to approve the pending transfer. In other words, by calling $\approveLeader()$, Bob gives up the right to contest meaning Carol will not have to wait out the contest window in order to call $\replaceLeader()$. In that case, if Bob cooperates, the transfer would be finalized sooner. If Bob doesn't cooperate, the transfer would be finalized later but not stoppable if Alice and Carol are conforming.

This can be accomplished by adding to the $AB(BA)$ contracts an $approved$ flag indicating whether Bob has given up his ability to call $\contestLeader()$.

Phases II-III are the only ones affected.

We add the following instruction to the \emph{Consistency Phase}:
\begin{enumerate}
    \item Once Bob sees Alice has called $\mutateLockLeader()$ on both $AB$ and $BA$ with the same change, then before a $\Delta$ elapses, Bob should call $\approveLeader()$ on both $AB$ and $BA$.
\end{enumerate}
and to the \emph{Replace/Revert Phase}:
\begin{enumerate}
    \item If Carol sees Bob approved the tentative consistent changes on both $AB$ and $BA$, then before a $\Delta$ elapses, she calls $\replaceLeader()$ on both $AB$ and $BA$.
\end{enumerate}

If in addition to conforming to the base protocol in Section 5.1, Bob also executes the approve step shown previously, then we call Bob \emph{altruistic}.
This just indicates that Bob is willing to speed up the termination of the protocol even if it doesn't necessarily change his financial position.

\begin{figure}
        \centering
        \includegraphics[width = \linewidth]{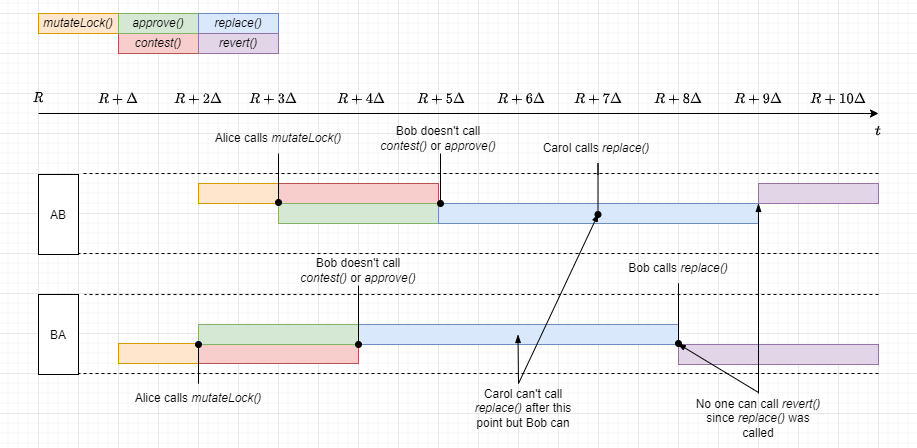}
        \caption{A protocol execution demonstrating why Bob needs two more \textit{replaceLeader()} rounds than Carol.
        If not, then he lacks full $\Delta$ to call \textit{replaceLeader()} on $BA$ after Carol calls it on $AB$.
        Here we assume Alice has reported consistent signatures in the \emph{Mutate Lock Phase} but has reported them a $\Delta$ apart.
        Colored blocks represent time periods when functions can be called.
        The \textit{Leader} suffix is excluded from the function calls for simplicity.
        } 
        \label{figure:execution}
\end{figure}

\subsection{Transfer Follower Position}
\label{transfer_follower}

We will now describe a protocol that allows Bob to transfer his position in a swap with Alice, to a third-party David.
The main difference with this protocol and the previous one, is that Bob cannot limit Alice's optionality.
Namely, he cannot unilaterally lock the asset in the $BA$ contract since this would limit the optionality Alice purchased.
Effectively Alice must always be able to claim the funds on $BA$ with her secret until the swap itself times out.
This weakening of the constraint that we had in the previous protocol, allows for a much simpler protocol.
It consists of 2 phases:
\begin{enumerate}
    \item \emph{Mutate Phase}
    \item \emph{Replace/Revert Phase}
\end{enumerate}

\paragraph{I: Mutate Phase}

\begin{enumerate}
    \item  Bob and David agree on $\Asset_D$ at time $\startFollower$, similar to the replacing leader protocol. David creates $DB$ contract, escrowing $\Asset_D$, setting the $\swapHashlock$ as $H(D_1)$ for Bob to redeem  $\Asset_D$ and the timeout be $\startFollower+5\Delta$. 
    \item  If the previous step succeeds, and Alice has not revealed $A_1$ before $\Delta$ elapses,
    Bob concurrently calls $\mutateLockFollower()$ on $AB$ and $\mutateLockFreeFollower()$ on $BA$ contracts, tentatively setting $[\replaceHashlock=H(D_1)]$ to both contracts, and  setting $\candidateSender = \David.\addressVar$ on $BA$ and $\candidateReceiver = \David.\addressVar$ on $AB$.
    We use $\mutateLockFreeFollower()$ because on the $AB$ contract,
    this mutation does not lock the asset on $AB$. 
    This is necessary to ensure Alice does not temporarily lose her optionality.
\end{enumerate}

\paragraph{II: Replace/Revert Lock Phase}

\begin{enumerate}
    \item If David sees Bob successfully call $\mutateLockFollower()$, $\mutateLockFreeFollower()$ on both $AB$ and $BA$ respectively, then within $\Delta$ he concurrently executes $\replaceFollower()$ with $D_1$ on $AB$ and $BA$ to replace Bob.
    \item If Bob sees David call $\replaceFollower()$ on either $AB$ or $BA$, within $\Delta$ he should call $\claim()$ on $DB$ to claim $\Asset_D$.
\end{enumerate}

From this point on, we will use Protocol \ref{transfer_leader} to refer the protocol in Section \ref{transfer_leader}, Protocol \ref{transfer_leader_altruistic} to the refer to the protocol \ref{transfer_leader_altruistic} with extra steps from Section \ref{transfer_leader_altruistic}, Protocol \ref{transfer_follower} to refer to the protocol in Section  \ref{transfer_follower},  and Protocol \ref{sec:multi_candidate} to refer to the protocol in Section  \ref{sec:multi_candidate}. 

\subsection{Handling Multiple Candidates}
\label{sec:multi_candidate}
\input{multiple_candidate}

%% file: multiple_candidate.tex
\label{multi_transfer_leader}

When Alice decides to tentatively transfer her swap option to Carol, it could be the case that Carol is intending to just grief Alice.
In the worst case, Alice will have her option locked for $6\Delta$ by Carol that simply doesn't participate in the \textit{Replace/Revert} Phase of Protocol \ref{transfer_leader}.

To combat this, Alice might want to initiate several concurrent transfers with multiple buyers at once.

A straightforward solution to this is, after Alice gets her position back from a failed transfer with $Carol_i$, Alice can call $mutateLockLeader()$ again to indicate she wants to transfer her option to a new option buyer $Carol_{i+1}$.
In the worst case, in this protocol, Alice would spend $6\Delta$ with her option locked for each potential buyer who doesn't comply.

One question then is, can we have Alice potentially transfer her position to multiple buyers at the same time, or with less waiting time than 6$\Delta$ in between each new buyer? 
Although we cannot let Alice transfer her option to multiple buyers simultaneously, we can create enough overlap between each of the potential to reduce Alice's waiting time for each potential buyer.
The protocol description below outlines how Alice is able to reduce her waiting time for each potential buyer to $4\Delta$ rather than 6$\Delta$.

The main idea of the protocol is to use to run a concurrent version of Protocol \ref{transfer_leader} for each potential buyer.
Alice assigns each potential buyer a ticket (\textit{sequence number}). A shared \textit{counter} is synchronized between the $AB$ and $BA$ contracts.
The shared \textit{counter} assigns an ordering to each of the potential buyers to execute their respective versions of Protocol \ref{transfer_leader}.
In this way, it serves as a first-come first-serve mechanism for Alice to sell her swap option.

We will use $Carol_i$ to denote the buyer who is assigned the $i$-th sequence number by Alice. 
On each contract $AB$ and $BA$, each potential buyer $Carol_i$, will have their own associated state structure $state_i$, similar to Protocol \ref{transfer_leader}.
This keeps track of the state relevant to the 3 different phases from Protocol \ref{transfer_leader}.

The protocol is defined as follows:
\begin{itemize}
    \item Initially \textit{counter} on \textit{AB} and \textit{BA} contracts is initialized to $0$.
    \item For each potential buyer, Alice assigns a unique sequence number \textit{seq} (starting from 0 and growing by 1 for each assignment) to them, which represents that party's position in the queue for a tentative replacement.
    If Alice assigns the same sequence number to different parties and sends both of them to the \textit{AB/BA} contracts, the conflict is resolved via the \textit{Consistency Phase} from the original Protocol \ref{transfer_leader}.
    It is the same as when Alice deviates by sending inconsistent mutation transactions to both contracts, in which case Bob contests and the tentative transfer is reverted.
    \item Alice sends mutation transactions as in Protocol \ref{transfer_leader}.
    The main difference is that the mutation transactions now include a sequence number \textit{seq} for the candidate replacement. 
    The $mutateLockLeader()$, $ contestLeader()$ functions all take this sequence number as parameters as part of the mutation.
    If $counter==seq$, the transaction is accepted. Otherwise, the transaction is rejected.
    
    If Alice is conforming, \textit{counter} should be synchronized on \textit{AB/BA} within $\Delta$.
    If a tentative transfer is in progress, a new mutation transaction with $seq==counter+1$ will be accepted only if $4 \Delta$ has elapsed after the first mutation. A mutate transaction with $seq==counter$ can serve as mutation transaction for a contest of the current tentative transfer. 
    After the tentative change is reverted, i.e.  \textit{revertLeader()} is called, the \textit{counter} is incremented by 1.  This design can be optimized by accepting the mutation transactions with larger sequence numbers and store them in a queue for future use.
    These can then take effect immediately after
    $4 \Delta$ has elapsed after the current tentative mutation happens and \textit{revertLeader()} has been called. The reason why $4\Delta$ is sufficient interval between the execution of base protocols for two potential buyers is that, after $4 \Delta$, there is only $2\Delta$ left for the mutation to be reverted if the previous potential buyer gives up. The later arriving buyer can use this window to finish its consistency phase and after it ends, its replace phase can start without waiting.
    
    \item A mutation transaction can take effect (The corresponding candidate can call \textit{replaceLeader()} ) only after the previous mutation transaction has expired, meaning it is reverted and Alice regains the position. Then after that the new candidate can execute Protocol \ref{transfer_leader} to replace Alice.
\end{itemize}

The protocol for transferring the follower position to multiple candidates is similar.

%% file: properties.tex
\subsection{General Transfer Properties}
\label{general_properties}
Our leader and follower transfer protocols overlap in some properties they satisfy.
We outline those properties here.

 Recall that assets escrowed in the contracts are called \emph{principal}. The principals involved in the swap option are always Alice's principal and Bob's principal. Here, the original \emph{option provider} is Bob (follower) and the \emph{option owner} is Alice (leader). In the leader transfer protocol and the follower transfer protocol, what is transferred is a position in the swap option.  If Alice transfers her position it is to Carol, and Bob to David.  We say one party owns a position in an option if they are \emph{option owner} (Alice/Carol):  one can release a secret, receiving Bob's principal by relinquishing Alice's, or let the option expire and Alice's principal is refunded to them, or \emph{option provider} (Bob/David): one provides the option owner the right, but not obligation, for the exchange of Alice's principal with Bob's.  

 From now on, we call a leader/follower who wants to transfer their position as position seller, and the one who wants to replace them as position buyer. 
\begin{itemize}
    \item  \textit{Liveness}: In a transfer of positions, if all parties are conforming, then leader/follower transfers their position to a buyer and the position seller gets proper payment from the buyer.
    
    \item \textit{Transfer independence}: A compliant position seller (leader/follower) transferring their position to another compliant position buyer, can successfully transfer their position without the cooperation of a third counterparty (follower/leader).

      \item \textit{Non-blocking transfer with adversarial counterparty}: If a compliant position seller (leader/follower) is transferring their position to another compliant position buyer, a third counterparty (follower/leader) cannot interfere with the transfer.
    
      \item \textit{Transfer atomicity}: If a compliant position seller (leader/follower) loses their position to the buyer, then they receive the expected principal from the position buyer.

     \item \textit{No UNDERWATER for a conforming party (Safety)}: \textit{UNDERWATER} means a party loses their outgoing principal without getting their incoming principal.
  No \textit{UNDERWATER} guarantees a conforming party's safety since it will never end up with losing their principal without getting principals from others.
\end{itemize}

\begin{theorem}
    \label{underwater_thm}
    Protocol \ref{transfer_leader} and \ref{transfer_follower} satisfy No UNDERWATER for Alice, Bob, Carol and David:
    \begin{itemize}
        \item If Alice is conforming, then if she loses her principal, she either gets Bob's or Carol's principal, or both.
        \item If Bob is conforming, then if he loses his principal, he gets Alice's principal or David's, or both.
        \item If Carol is conforming, then if she loses her principal,  she gets Alice's principal or Bob's principal, or both.
        \item If David is conforming, then if he loses his principal, he gets Alice's principal or Bob's principal, or both.
    \end{itemize}
\end{theorem}

\begin{proof}
See Appendix \ref{underwater_proofs}.
\end{proof}

\subsection{Leader Transfer Properties}
\label{leader_properties}

We first demonstrate how our leader transfer protocol satisfies the general transfer properties from \ref{general_properties}.

Because the detailed proofs of the following results are long and similar to those that we will see for the follower transfer protocol \ref{follower_properties}, we omit them here.
See the Appendix \secref{proof} for details.

\begin{restatable}{theorem}{leaderLivenessThm}
\label{leader_liveness_thm}
Protocol \ref{transfer_leader} satisfies liveness: If Alice, Bob, and Carol are all conforming, then Alice gets Carol's principal, Carol gets Alice's position, and Bob maintains his position.
\end{restatable}

\begin{restatable}{theorem}{leaderTransferIndThm}
\label{leader_transfer_ind_thm}
Protocol \ref{transfer_leader} satisfies transfer independence: Alice can transfer her position to Carol without Bob's participation.
\end{restatable}

\begin{restatable}{theorem}{leaderNonblockThm}
\label{leader_nonblock_thm}
Protocol \ref{transfer_leader} satisfies non-blocking transfer: Alice can transfer her position to Carol even if Bob is adversarial.
\end{restatable}

\begin{restatable}{theorem}{leaderTransferAtomThm}
\label{leader_transfer_atom_thm}
Protocol \ref{transfer_leader} satisfies transfer atomicity: If Alice loses her position in the swap, then she can claim Carol's principal.
\end{restatable}

In addition to the properties from Section \ref{general_properties}, it is desirable for a leader transfer protocol to have one additional property.

Given two protocols $P,P'$ that satisfy \textit{non-blocking transfer with adversarial counterparty (Bob)}, Bob is called \textit{altruistic} in $P'$ with respect to $P$ if Bob conforming in $P'$ terminates faster than Bob conforming in $P$ in executions where all parties are conforming.
Importantly, it is not necessary that Bob should be incentivized to choose to follow $P'$ over $P$.

When  Alice and Carol are conforming, even though adversarial Bob cannot block the transfer, Bob can delay the transfer for a few rounds.
The following next property we introduce is a slightly stronger version of \textit{non-blocking} property that is ideal for a leader transfer protocol.

\begin{enumerate}
    \item \textit{Timely transfer with altruistic Bob}: A protocol $P'$ satisfies this property if there exists a protocol $P$ satisfying \textit{non-blocking transfer with adversarial counterparty (Bob)} where Bob is \textit{altruistic} in $P'$ with respect to $P$.
\end{enumerate}

\begin{theorem}
    Protocol \ref{transfer_leader_altruistic} satisfies timely transfer with altruistic Bob.
\end{theorem}

\begin{proof}
Since Protocol \ref{transfer_leader} is \textit{non-blocking} by Theorem \ref{leader_nonblock_thm},
it is enough to show that Bob is altruistic in Protocol \ref{transfer_leader_altruistic} with respect to Protocol \ref{transfer_leader}.
By assumption all parties are conforming.
In both protocols, Carol will have created $MutSwapCA$ by $startLeader + \Delta$ and Alice will have called $mutateLockLeader()$ on $AB$ and $BA$ by $startLeader + 2\Delta$.
In Protocol \ref{transfer_leader_altruistic} Bob now calls $approveLeader()$ on $MutSwapAB$ and $MutSwapBA$ by $startLeader + 3\Delta$.
Thus the Replace/Revert Phase for Protocol $\ref{transfer_leader_altruistic}$ begins because Bob skips the contest phase.
In Protocol \ref{transfer_leader}, since Alice is compliant, her signatures on $MutSwapAB$ and $MutSwapBA$ are consistent.
Thus Bob won't call $contest()$ on $MutSwapAB$ or $MutSwapBA$ during this $2\Delta$ Consistency Phase.
Carol waits $2\Delta$ after $startLeader + 2\Delta$.
In this case, the Replace/Revert Phase begins at $startLeader + 4\Delta$.
Thus Bob is altruistic in Protocol \ref{transfer_leader_altruistic} with respect to Protocol \ref{transfer_leader}.
\end{proof}

\subsection{Follower Transfer Properties}
\label{follower_properties}

We now demonstrate how our follower transfer protocol also satisfies the general transfer properties from Section \ref{general_properties}.

Again, the proofs are included in the Appendix \secref{proof}.

\begin{restatable}{theorem}{followerLivenessThm}
\label{follower_liveness_thm}
Protocol \ref{transfer_follower} satisfies liveness: If Alice, Bob, and David are all conforming, and Alice doesn't reveal $A_1$, then Bob gets David's principal, David gets Bob's position, and Alice maintains her position.
\end{restatable}

\begin{restatable}{theorem}{followerTransferIndThm}
\label{follower_transfer_ind_thm}
Protocol \ref{transfer_follower} satisfies transfer independence: Bob can transfer his position to David without Alice's participation.
\end{restatable}

\begin{restatable}{theorem}{followerNonblockThm}
\label{follower_nonblock_thm}
Protocol \ref{transfer_follower} satisfies non-blocking transfer: Bob can transfer his position to David even if Alice is adversarial.
\end{restatable}

Note that Alice is not adversarial by choosing to reveal her secret $A_1$ since this just means she is exercising her option.
This is consistent with how she should behave in the original swap protocol.

\begin{restatable}{theorem}{followerTransferAtomThm}
\label{follower_transfer_atom_thm}
Protocol \ref{transfer_follower} satisfies transfer atomicity: If Bob loses his swap position, then he can claim David's principal.
\end{restatable}

On top of the base properties described in Section \ref{general_properties}, any protocol transferring  Bob's position should also be:

\begin{enumerate}
    \item \emph{Optionality preserving}: If Alice is compliant, she never loses her ability to exercise the original swap option.
\end{enumerate}

Namely, it is unfair for Alice to temporarily lose her right to use her option without her consent.

\begin{theorem}
Protocol \ref{transfer_follower} is optionality preserving.
\end{theorem}

\begin{proof}
It's enough to show that Alice can always claim the principal on $BA$, if she reveals her secret $A_1$ by timeout $T$.
The added functions from Protocol \ref{transfer_follower}, namely $mutateLockFreeFollower()$,$replaceFollower()$, and $revertFollower()$ only modify the state of $follower\_mutation$.
However, any call to $claim()$ on $BA$ doesn't depend on the state of $follower\_mutation$.
Namely, any call to $claim()$ is independent of the state of the follower transfer protocol. 
So its enough that Alice reveal $A_1$ before $T$.
\end{proof}

There is nothing preventing both Protocol \ref{transfer_leader} and Protocol \ref{transfer_follower} to be run concurrently.
That is, if Alice wants to transfer her position to Carol, and Bob wants to transfer his position to David, they can both do so simultaneously.
This follows immediately from the fact that the parts of the contracts managing the state for each individual protocol are entirely disjoint and cannot affect each other.

\subsection{Leader Transfer (Multiple Buyers) Properties}

Besides the general properties mentioned above, the protocol that handles multiple candidate buyers (Protocol \ref{multi_transfer_leader}) satisfies the following unique properties.

\begin{itemize}
    \item First-come first-serve (FCFS): If Alice is conforming (the buyer who tentatively pays to Alice earlier gets a smaller sequence number), then the earlier arriving buyer has priority to replace Alice's position.
    \item Starvation freedom: If Alice is conforming and a conforming buyer, say $Carol_j$, gets a sequence number $j$, then  $Carol_j$ can replace Alice's position if she is conforming and all $Carol_i$ where $i<j$ gives up the replacement.
\end{itemize}

\begin{theorem}
\label{thm:counter_sync}
If Alice is conforming, the counter on both contracts are synchronized with inconsistency for at most $\Delta$.
\end{theorem}

\begin{proof}
The counter is initialized as 0 and it incremented only after  \textit{revertLeader()} is enabled. The \textit{revertLeader()} is enabled only when $6\Delta$ elapses after the start time of  \textit{mutateLockLeader()} being called. Since the start time of \textit{mutateLock()} on both contracts are staggered by at most $\Delta$, when the counter is incremented on one contract, it can be incremented on another contract within $\Delta$.
\end{proof}

\begin{theorem}
If Alice is conforming, Protocol \ref{multi_transfer_leader} satisfies FCFS.
\end{theorem}

\begin{proof}
\label{proof:fcfs}
If Alice is conforming, she issues \textit{mutateLockLeader()} transactions with ascending sequence numbers to different potential buyers, in an arrive-earlier-smaller-sequence manner. For every potential buyer, a base transfer protocol (Protocol  \ref{transfer_leader}) is run as normal.  Without loss of generality, at time $t$, suppose on AB contract the counter is $i$ and on BA the counter is $j=i-1$. On AB, $Carol_i$ has priority. By Theorem \ref{thm:counter_sync}, we know the counter on $BA$ will become $i$ within $\Delta$ and then $Carol_i$ has priority. The reason why the counter is incremented from $j$ to $i$ is that $Carol_j$ gives up the replacement. Then we see on both contracts $Carol_i$ has priority. It is obvious $Carol_k$ where $k\geq i+1$ cannot start their replacement unless $Carol_i$ gives up, i.e. the replacement phase ends and the mutation is reverted, since the counter is $i$ now.  Therefore the protocol satisfies FCFS.
\end{proof}

\begin{theorem}
If Alice is conforming, Protocol \ref{multi_transfer_leader} is starvation free.
\end{theorem}

\begin{proof}
Similar to Proof \ref{proof:fcfs},  if Alice is conforming and a buyer  $Carol_j$, gets a sequence number $j$, then the base protocol for $Carol_j$ can start underlying others' base protocols. If all $Carol_i$ where $i<j$ gives up the replacement, then the counter is incremented to $j$ and then $Carol_j$ can call \textit{replaceLeader()} to replace Alice.
\end{proof}

%% file: related.tex
Cross-chain options have their origin in cross-chain atomic swap protocols.
A cross-chain atomic swap enables two parties to exchange assets across different blockchains.
An atomic swap is implemented via hashed timelock contracts(HTLC)\cite{tiernolan}.
There are a variety of protocols proposed \cite{decred2,atomicdex,sparkswap,altcoin} and implemented \cite{bip199s,decredgit}.
Herlihy \emph{et al.} proposed protocols for atomic multi-party swaps \cite{Herlihy2018}
and more general atomic cross-chain deals \cite{HerlihyLS2019}. 

Several researchers~\cite{han2019optionality,arwen2019,liu2018,ZmnSCPxj} have noted
that most two-party swap protocols effectively act as poorly-designed
options~\cite{higham2004introduction},
because one party has the power to decide whether to
go through with the agreed-upon swap without compensation for its counterparty.

A number of proposals~\cite{han2019optionality,arwen2019,liu2018,ZmnSCPxj,eizinger,komodo,XueH2021}
address the problem of optionality in cross-chain atomic swaps by introducing
some form of premium payment,
where a party that chooses not to complete the swap pays a premium to the counterparty.
Robinson~\cite{danrobinson} proposes to reduce the influence of optionality
by splitting each swap into a sequence of very small swaps. 
Han \emph{et al.} \cite{han2019optionality} quantified optionality unfairness in atomic swap using the Cox-Ross-Rubinstein option pricing model \cite{cox1979option}, treating the atomic swap as an American-style option. The Black-Scholes (BS) Model \cite{black2019pricing} can be used to estimate the value of European-style options. 

Liu~\cite{liu2018} proposed an alternative approach where option providers are
are paid up-front for providing optionality,
as in the conventional options market.
In this protocol,
Alice explicitly purchases an option from Bob by paying him an nonrefundable premium.
Tefagh \emph{et al.}~\cite{tefaghcapital} proposed a similar protocol
which enables Alice to deposit her principal later than Bob.
None of these works have considered how to close an option owner's position
by transferring that option to a third party.

There are protocols that allow blockchains to communicate (cross-chain proofs), however they either rely on external third parties ~\cite{goes2020interblockchain} or their applicability requires the introduction of centralized services, modifications to existing software, and doesn't guarantee reliable message delivery ~\cite{robinson2021survey}.

%% file: conclusions.tex
The transferable swap protocols presented here have certain limitations.
As described earlier,
the protocols require multiple time-out periods.
If all parties are responsive,
these timeouts should not affect the performance of normal executions,
but they could lead to long worst-case executions.
The protocols also include a ``hard timeout''
where Alice and Carol pause to give Bob a chance to
object to a malformed transfer.
Bob can be paid an incentive to respond quickly,
but he could slow down (but not stop) the protocols execution.

The ability for Bob to be able to report inconsistent state changes between multiple blockchains is integral to the design of our protocols.
Adding such functionality was not necessary in simpler cross-chain protocols like the two-party swap.
In future work, we hope to better understand how the complexity of a cross-chain deal relates to the necessity of this consistency phase.

%% file: code.tex
\subsection{Contracts}\label{sec:code}

\begin{lstlisting}[language=Solidity]

//Assume the following:

// now is the time the transaction is included in a block 
// A clear() function on the Mutation struct that resets all fields to default values
// A hash function H
// || denotes concatenation of inputs to a hash function
// A sig (digital signature) object with the following functions
//  valid(address) -> bool, returns true if valid signature by address 
//  msg() -> [hash], returns an array of hashlocks if the signature signed such a message, null otherwise 

contract MutSwapAB{
   
    struct FollowerMutation{
        //Signature by follower to allow candidate to take its position
        Sig voucher;
        //Candidate party to replace follower
        address candidate_receiver;
        //Hashlock used to replace follower
        uint replace_hash_lock;
        //Time mutation begins
        uint start_time;
        //Flag for freezing asset when a tentative replacement is happening
        bool mutating;
        //Controls whether asset is locked or not during mutation
        bool can_lock_asset;
    }

    struct LeaderMutation{
        //Signature by follower to allow candidate to take its position
        Sig voucher;
        //Candidate party to replace follower
        address candidate_sender;
        //Party who called mutateLockLeader
        address mutator;
        //Hashlock used to replace follower
        uint replace_hash_lock;
        //Hashlock for exercising option
        uint swap_hash_lock;
        //Time mutation begins
        uint start_time;
        //Used for optimistic execution of protocol
        bool approved;
        //Flag for freezing asset when a tentative replacement is happening
        bool mutating;
        //Controls whether asset is locked or not during mutation
        bool can_lock_asset;
    }

    //State information for base swap protocol
    Asset asset; //Reference to preferred token contract
    address sender; //Current sender of escrowed funds
    address receiver; //Current receiver of escrowed funds
    address leader; //Leader of the original swap protocol
    address follower; //Follower of the original swap protocol
    uint swap_hash_lock; //Hashlock of the swap protocol 
    uint T_AB; //Timeout for locked asset
    uint delta = 10 minutes; //Assumed worst case transaction inclusion time (Is arbitrary)

    //State info associated with mutable leader position
    LeaderMutation leader_mutation;

    //State info associated with mutable follower position
    FollowerMutation follower_mutation; 

    //Controls whether follower/leader positions can be changed
    bool mutable;
   
    function MutSwapAB(Asset _asset,uint start,uint dT,address _sender,address _receiver,uint _swap_hash_lock,address _leader,address _follower,bool _mutable){
        require(msg.sender == _sender); //Sender can only escrow their own funds
        require(leader != follower); //Leader and follower should be distinct
        require(leader == _sender);
        require(follower == _receiver);

        this.mutable = _mutable;

        //Inital mutation states
        leader_mutation.mutating = false; 
        leader_mutation.approved = false;
        follower_mutation.mutating = false; 

        sender = _sender;
        receiver = _receiver;
        swap_hash_lock = _swap_hash_lock; //Hashlock for the initial swap
        leader = _leader;
        follower = _follower;
        asset = _asset;
        asset.send(address(this));
        T_AB = start + (dT+1)*delta; //Sender is Alice

        if(mutable){
            //For AB contract, asset is actively locked during leader and follower transfer
            leader_mutation.can_lock_asset = true;
            follower_mutation.can_lock_asset = true;
        }
    }
   
    function claim(string secret){

        //If a previous mutation lock was never completed, revert to original swap
        if(mutable){
            if(leader_mutation.mutating){
                revertLeader();
            }
            if(follower_mutation.mutating){
                revertFollower();
            }
        }

        require(msg.sender == receiver.id); //Only receiver can call claim
        require(now <= T_AB); //Must be before timeout
        require(H(secret) == swap_hash_lock); //Claim conditional on revealing secret
        asset.send(receiver);
    }
   
    function refund(){
        
        //If a previous mutation lock was never completed, revert to original swap
        if(mutable){
            if(leader_mutation.mutating){
                revertLeader();
            }
            if(follower_mutation.mutating){
                revertFollower();
            }
        }

        require(msg.sender == sender); //Only sender can call refund
        require(now > T_AB); //After lock has timed out
        asset.send(sender);
    }

    function mutateLockLeader(Sig sig,address _candidate_sender, uint _replace_hash_lock,uint _swap_hash_lock){
        require(mutable);
        
        //If a previous mutation lock is stale, then call revert to allow for a new mutation lock to be made
        if(leader_mutation.mutating){
            revertLeader();
        }
        
        require(!leader_mutation.mutating); //Only one mutate_lock at a time 
        require(msg.sender == sender || msg.sender == receiver); //Only sender or receiver can mutate
        if(msg.sender == leader){ //Alice has less time to call mutateLock
            require(T_AB >= now + 8*delta);
        }else if(msg.sender == follower){
            //Bob is given more time to call mutateLock in response to Alice
            require(T_AB >= now + 7*delta);
        }
        
        require(sig.valid(leader)); //Requres Alice's sig of Carol's hashlocks
        require(sig.msg() == [_replace_hash_lock,_swap_hash_lock]); //The msg has to just be the candidate's hashlocks
        leader_mutating.mutating = true;
        leader_mutating.voucher = sig; //Proof that candidate was approved  
        leader_mutating.candidate_sender = _candidate; 
        leader_mutating.mutator = msg.sender;
        leader_mutating.replace_hash_lock = _replace_hash_lock;
        leader_mutating.swap_hash_lock = _swap_hash_lock;
        leader_mutating.start_time = now; //Used for flexible timeouts in the transfer
    }
   
    function mutateLockFollower(Sig sig,address _candidate_receiver, uint _replace_hash_lock){
        require(mutable);

        //If a previous mutation lock is stale, then call revert to allow for a new mutation lock to be made
        if(follower_mutation.mutating){
            revertFollower();
        }

        require(msg.sender == follower); //Only Bob can call
        require(!follower_mutation.mutating);
        require(follower_mutation.lock_asset); //Only can call mutateLockFollower on AB contract
        require(sig.valid(follower)); //Bob's valid signature
        require(sig.msg() == [_replace_hash_lock]);
        require(T_AB >= now + 3*delta); //Need enough time for David to call claim        

        follower_mutation.mutating = true;
        follower_mutation.voucher = sig; //Proof that candidate was approved by previous owner of the position
        follower_mutation.candidate_receiver = _candidate_receiver;
        follower_mutation.replace_hash_lock = _replace_hash_lock;
        follower_mutation.start_time = now; //Used for flexible timeouts in the transfer
    }

     function replaceLeader(string secret){
        require(mutable);
        require(leader_mutation.mutating);
        bool candidate_round = false;
        bool follower_round = false;
        
        //If Bob mutated, can skip contesting phases since he is implicitly approving of the signature he called mutateLock with
        if(leader_mutation.mutator == follower || (leader_mutation.mutator == leader && leader_mutation.approved)){ //Can immediately call replace once Bob has mutated 
            //Carol can replace once Bob has called mutate 
            candidate_round = (msg.sender == leader_mutation.candidate_sender) && (4*delta >= now - leader_mutation.start_time > 0);
            //Bob can call replace 
            follower_round = (msg.sender == follower) && (6*delta >= now - leader_mutation.start_time > 0);

        }
        else{ //Contesting (pessimistic case)
            //Carol can replace in round right after last chance at contesting
            candidate_round = (msg.sender == leader_mutation.candidate_sender) && (4*delta >= now - leader_mutation.start_time > 2*delta);
            //Bob can replace in the 3 rounds right after last chance at mutating round
            follower_round = (msg.sender == follower) && (6*delta >= now - leader_mutation.start_time>0);
        }

        require(candidate_round || follower_round);
        require(H(secret) == leader_mutation.replace_hash_lock);
       
        sender = leader_mutation.candidate_sender; //Carol takes refund optionality from Alice
        swap_hash_lock = mutation.swap_hash_lock;
       
        leader_mutation.approved = false;
        leader_mutation.mutating = false;
        leader_mutation.mutation.clear();
    }
   
    //If no contesting has occurred, transfer can be complete
    function replaceFollower(string secret){
        require(mutable);
        require(follower_mutation.mutating);
        require(2*delta >= now - follower_mutation.start_time > 0); //2 rounds given for replacement
        
        require(H(secret) == follower_mutation.replace_hash_lock);

        receiever = follower_mutatation.candidate_receiver; //David takes refund optionality from Bob

        follower_mutatation.clear();
        follower_mutation.mutating = false;
    }

     //Contest if transferring party tries to transfer to two different parties simultaneously
    //Note this method only does something if it can be proved that Alice lied
    //This method can't do anything if Alice is honest
    function contestLeader(Sig sig,string secret){
        require(mutable);
        require(msg.sender == follower); //Only allow Bob to call this
        require(leader_mutation.mutating); //Can only contest a mutation if one has occured
        require(!leader_mutation.approved); //Cannot contest after approval 
        require(leader_mutation.mutator == leader); //Can only contest a mutation when Alice called it, not Bob
        require(2*delta >= now - leader_mutation.start_time > 0); //Can contest only in the 2 rounds after mutation
        //Can contest if Alice creates two inconsistent signatures or tries to reveal preimage of hashlock too early
        require(sig.valid(leader) || H(secret) == swap_hash_lock); //Make sure Alice actually signed the sig
        if(sig != leader_mutation.voucher || H(secret) == swap_hash_lock){ //Checks if Alice reported inconsistent sigs
            leader_mutation.clear();
            leader_mutation.mutating = false;
        }
       
    }
   
    //Issue is really in the sequential consistent case
    //Bob can call approve in any next two rounds after a mutation
    function approveLeader(){
       require(mutable);
       require(msg.sender == follower); //Only Bob can approve 
       require(leader_mutation.mutating); //Can only approve after a mutation has begun 
       require(leader_mutation.mutator == leader); //Can only call approve when Alice was the one who called mutate 
       require(2*delta >= now - leader_mutation.start_time > 0); //Can approve only in the 2 rounds after mutation
       leader_mutation.approved = true;
    }

    //If replacement doesn't happen quickly enough, can revert back to the old swap
    //Isn't called directly in order to save an extra round of the protocol
    function revertLeader(){
        require(mutable);
        require(leader_mutation.mutating);
        require(now - leader_mutation.start_time > 6*delta);

        leader_mutation.approved = false;
        leader_mutation.mutating = false;
        leader_mutation.clear();
    }
   
    //If replacement doesn't happen quickly enough, can revert back to the old swap
    //Isn't called directly in order to save an extra round of the protocol
    function revertFollower(){
        require(mutable);
        require(follower_mutation.mutating);
        require(now - follower_mutation.start_time > 2*delta);

        follower_mutation.mutating = false;
        follower_mutation.clear();
    }
   
}

\end{lstlisting}

\begin{lstlisting}[language=Solidity]

//Assume the following:

// now is the time the transaction is included in a block 
// A clear() function on the Mutation struct that resets all fields to default values
// A hash function H
// || denotes concatenation of inputs to a hash function
// A sig (digital signature) object with the following functions
//  valid(address) -> bool, returns true if valid signature by address 
//  msg() -> [hash], returns an array of hashlocks if the signature signed such a message, null otherwise 

contract MutSwapBA{
   
    struct FollowerMutation{
        //Signature by follower to allow candidate to take its position
        Sig voucher;
        //Candidate party to replace follower
        address candidate_sender;
        //Hashlock used to replace follower
        uint replace_hash_lock;
        //Time mutation begins
        uint start_time;
        //Flag for freezing asset when a tentative replacement is happening
        bool mutating;
        //Controls whether asset is locked or not during mutation
        bool can_lock_asset;
    }

    struct LeaderMutation{
        //Signature by follower to allow candidate to take its position
        Sig voucher;
        //Candidate party to replace follower
        address candidate_receiver;
        //Party who called mutateLockLeader
        address mutator;
        //Hashlock used to replace follower
        uint replace_hash_lock;
        //Hashlock for exercising option
        uint swap_hash_lock;
        //Time mutation begins
        uint start_time;
        //Used for optimistic execution of protocol
        bool approved;
        //Flag for freezing asset when a tentative replacement is happening
        bool mutating;
        //Controls whether asset is locked or not during mutation
        bool can_lock_asset;
    }

    //State information for base swap protocol
    Asset asset; //Reference to preferred token contract
    address sender; //Current sender of escrowed funds
    address receiver; //Current receiver of escrowed funds
    address leader; //Leader of the original swap protocol
    address follower; //Follower of the original swap protocol
    uint swap_hash_lock; //Hashlock of the swap protocol 
    uint T_BA; //Timeout for locked asset
    uint delta = 10 minutes; //Assumed worst case transaction inclusion time

    //State info associated with mutable leader position
    LeaderMutation leader_mutation;

    //State info associated with mutable follower position
    FollowerMutation follower_mutation; 

    //Controls whether follower/leader positions can be changed
    bool mutable;
   
    function MutSwapBA(Asset _asset,uint start,uint dT,address _sender,address _receiver,uint _swap_hash_lock,address _leader,address _follower,bool _mutable){
        require(msg.sender == _sender); //Sender can only escrow their own funds
        require(leader != follower); //Leader and follower should be distinct
        require(follower == _sender);
        require(leader == _receiver);

        this.mutable = _mutable;

        //Inital mutation states
        leader_mutation.mutating = false; 
        leader_mutation.approved = false;
        follower_mutation.mutating = false; 

        sender = _sender;
        receiver = _receiver;
        swap_hash_lock = _swap_hash_lock; //Hashlock for the initial swap
        leader = _leader;
        follower = _follower;
        asset = _asset;
        asset.send(address(this));
        T_BA = start + dT*delta;

        if(mutable){
            leader_mutation.lock_asset = true;
            follower_mutatation.lock_asset = false; //Prevents Bob from denying Alice her optionality by locking her asset
        }
    }
   
    function claim(string secret){

        //If a previous mutation lock was never completed, revert to original swap
        if(mutable){
            if(leader_mutation.mutating){
                revertLeader();
            }
        }

        require(msg.sender == receiver.id); //Only receiver can call claim
        require(now <= T_BA); //Must be before timeout
        require(H(secret) == swap_hash_lock); //Claim conditional on revealing secret
        asset.send(receiver);
    }
   
    function refund(){
        
        //If a previous mutation lock was never completed, revert to original swap
        if(mutable){
            if(leader_mutation.mutating){
                revertLeader();
            }
            if(follower_mutation.mutating){
                revertFollower();
            }
        }

        require(msg.sender == sender); //Only sender can call refund
        require(now > T_BA); //After lock has timed out
        asset.send(sender);
    }

    function mutateLockLeader(Sig sig,address _candidate_receiver, uint _replace_hash_lock,uint _swap_hash_lock){
        require(mutable);
        
        //If a previous mutation lock is stale, then call revert to allow for a new mutation lock to be made
        if(leader_mutation.mutating){
            revertLeader();
        }
        
        require(!leader_mutation.mutating); //Only one mutate_lock at a time 
        require(msg.sender == sender || msg.sender == receiver); //Only sender or receiver can mutate
        if(msg.sender == leader){ //Alice has less time to call mutateLock
            require(T_BA >= now + 7*delta);
        }else if(msg.sender == follower){
            //Bob is given more time to call mutateLock in response to Alice
            require(T_BA >= now + 6*delta);
        }
        
        require(sig.valid(leader)); //Requres Alice's sig of Carol's hashlocks
        require(sig.msg() == [_replace_hash_lock,_swap_hash_lock]); //The msg has to just be the candidate's hashlocks
        leader_mutating.mutating = true;
        leader_mutating.voucher = sig; //Proof that candidate was approved  
        leader_mutating.candidate = _candidate_receiver;
        leader_mutating.mutator = msg.sender;
        leader_mutating.replace_hash_lock = _replace_hash_lock;
        leader_mutating.swap_hash_lock = _swap_hash_lock;
        leader_mutating.start_time = now; //Used for flexible timeouts in the transfer
    }

    //Provides similar functionality to mutateLockLeader but doesn't lock the escrowed resource from being locked
    function mutateLockFreeFollower(Sig sig,address _candidate_sender,uint _replace_hash_lock){
        require(mutable);

        require(msg.sender == follower); //Only Bob can call attest
        require(!follower_mutation.mutating);
        require(!follower_mutation.lock_asset);
        require(sig.valid(follower)); //Bob's valid signature
        require(sig.msg() == [_replace_hash_lock]);
        require(T_BA >= now + 2*delta); //Need enough time for David to call claim 

        follower_mutation.mutating = true;
        follower_mutation.voucher = sig; //Proof that candidate was approved  
        follower_mutation.candidate_receiver = _candidate_sender;
        follower_mutation.replace_hash_lock = _replace_hash_lock;
        follower_mutation.start_time = now; //Used for flexible timeouts in the transfer
    }

     function replaceLeader(string secret){
        require(mutable);
        require(leader_mutation.mutating);
        bool candidate_round = false;
        bool follower_round = false;
        
        //If Bob mutated, can skip contesting phases since he is implicitly approving of the signature he called mutateLock with
        if(leader_mutation.mutator == follower || (leader_mutation.mutator == leader && leader_mutation.approved)){ //Can immediately call replace once Bob has mutated 
            //Carol can replace once Bob has called mutate 
            candidate_round = (msg.sender == leader_mutation.candidate_receiver) && (4*delta >= now - leader_mutation.start_time > 0);
            //Bob can call replace 
            follower_round = (msg.sender == follower) && (6*delta >= now - leader_mutation.start_time > 0);

        }
        else{ //Contesting (pessimistic case)
            //Carol can replace in round right after last chance at contesting
            candidate_round = (msg.sender == leader_mutation.candidate_receiver) && (4*delta >= now - leader_mutation.start_time > 2*delta);
            //Bob can replace in the 3 rounds right after last chance at mutating round
            follower_round = (msg.sender == follower) && (6*delta >= now - leader_mutation.start_time > 0);
        }

        require(candidate_round || follower_round);
        require(H(secret) == leader_mutation.replace_hash_lock);
       
        receiever = leader_mutation.candidate_receiver; //Carol takes receiving optionality from Alice
        swap_hash_lock = mutation.swap_hash_lock;
       
        leader_mutation.approved = false;
        leader_mutation.mutating = false;
        leader_mutation.mutation.clear();
    }
   
    //If no contesting has occurred, transfer can be complete
    function replaceFollower(string secret){
        require(mutable);
        require(follower_mutation.mutating);
        require(2*delta >= now - follower_mutation.start_time > 0); //2 rounds given for replacement
        
        require(H(secret) == follower_mutation.replace_hash_lock);

        sender = follower_mutatation.candidate_sender; //David takes refund optionality from Bob
        
        follower_mutatation.clear();
        follower_mutation.mutating = false;
    }

     //Contest if transferring party tries to transfer to two different parties simultaneously
    //Note this method only does something if it can be proved that Alice lied
    //This method can't do anything if Alice is honest
    function contestLeader(Sig sig,string secret){
        require(mutable);
        require(msg.sender == follower); //Only allow Bob to call this
        require(leader_mutation.mutating); //Can only contest a mutation if one has occured
        require(!leader_mutation.approved); //Cannot contest after approval 
        require(leader_mutation.mutator == leader); //Can only contest a mutation when Alice called it, not Bob
        require(2*delta >= now - leader_mutation.start_time > 0); //Can contest only in the 2 rounds after mutation
        //Can contest if Alice creates two inconsistent signatures or tries to reveal preimage of hashlock too early
        require(sig.valid(leader) || H(secret) == swap_hash_lock); 
        if(sig != leader_mutation.voucher || H(secret) == swap_hash_lock){ //Checks if Alice reported inconsistent sigs
            leader_mutation.clear();
            leader_mutation.mutating = false;
        }
       
    }
   
    //Issue is really in the sequential consistent case
    //Bob can call approve in any next two rounds after a mutation
    function approveLeader(){
        require(mutable);
       require(msg.sender == follower); //Only Bob can approve 
       require(leader_mutation.mutating); //Can only approve after a mutation has begun 
       require(leader_mutation.mutator == leader); //Can only call approve when Alice was the one who called mutate 
       require(2*delta >= now - leader_mutation.start_time > 0); //Can approve only in the 2 rounds after mutation
       leader_mutation.approved = true;
    }

    //If replacement doesn't happen quickly enough, can revert back to the old swap
    //Isn't called directly in order to save an extra round of the protocol
    function revertLeader(){
        require(mutable);
        require(leader_mutation.mutating);
        require(now - leader_mutation.start_time > 6*delta);

        leader_mutation.approved = false;
        leader_mutation.mutating = false;
        leader_mutation.clear();
    }
   
    //If replacement doesn't happen quickly enough, can revert back to the old swap
    //Isn't called directly in order to save an extra round of the protocol
    function revertFollower(){
        require(mutable);
        require(follower_mutation.mutating);
        require(now - follower_mutation.start_time > 2*delta);

        follower_mutation.mutating = false;
        follower_mutation.clear();
    }
   
}

\end{lstlisting}

%% file: proof.tex
\subsection{Proofs}
\seclabel{proof}

\subsubsection{Misc Proofs}

First, we provide proofs for the properties of the leader transfer protocol.

\begin{lemma}
\label{mutate_lock}
If Alice and Carol are both conforming in the Mutate Lock Phase,  then both $AB$ and $BA$ are mutate locked with $(H(C_1),H(C_2))$ by $startLeader + 2\Delta$.
\end{lemma}

\begin{proof}
If Carol is conforming, she escrows her principal to Alice by time $R + \Delta$. Alice, after observing the creation of $CA$, will then call \textit{mutateLockLeader()} within a $\Delta$ on both $AB$ and $BA$, passing $(H(C_1),H(C_2))$ as the transfer hashlock and the new swap hashlock.
For any contract, if it is not mutate locked, then the \textit{mutateLockLeader()} sent by Alice can be included in the blockchain by $startLeader+2\Delta$.
If the \textit{mutateLockLeader()} transaction sent by Alice can not be included in the blockchain due to someone else calling it first, it must be that Bob called \textit{mutateLockLeader()} on the contract with $(H(C_1),H(C_2))$ before her. In any case, within $startLeader+2\Delta$, both contracts are mutate locked by $(H(C_1),H(C_2))$.
\end{proof}

If there is a successfully call to \textit{mutateLockLeader()} on $AB$ with $(H(C_1),H(C_2))$, let the first time that this transaction is called be $t_{AB}$.
Similarly, let $t_{BA}$ be the first time \textit{mutateLockLeader()} is called on $BA$ with $(H(C_1),H(C_2))$.

\leaderLivenessThm*
\begin{proof}

If Alice starts the transfer procedure, and all parties are conforming, then Carol owns the option and paid her principal to Alice. 
The argument is as follows.

By Lemma \ref{mutate_lock}, we know \textit{mutateLockLeader()} has been called by $R + 2\Delta$ on $AB$ and $BA$ with $(H(C_1),H(C_2))$.
After $\max\{t_{AB},t_{BA}\}+2\Delta \leq startLeader + 4\Delta$, Carol will know that Bob cannot successfully call \textit{contest()} on either $AB$ or $BA$.
Carol then releases her secret $C_1$ to $AB$ and $BA$ by $\max\{t_{AB},t_{BA}\}+3\Delta \leq startLeader + 5\Delta$.  
Then $AB$ and $BA$ are transferred to a new state where the swap hashlock is $H(C_2)$, the sender of $AB$ contract becomes Carol, and the receiver of $BA$ contract becomes Carol.
At this point, these contracts can be more appropriately renamed as as $CB$ and $BC$ respectively.
That is to say, Carol has set up a swap with Bob using her swap hashlock $H(C_2)$. 
Carol now owns the option.
Because the timeout on $CA$ is $startLeader + 10\Delta$, Alice can use $C_1$ to redeem Carol's principal on $CA$ by revealing $C_1$ by $startLeader + 6\Delta$.
.
\end{proof}

\leaderTransferIndThm*
\begin{proof}
If Alice and Carol are both conforming, then by Lemma \ref{mutate_lock}, we know  \textit{mutateLock()} will be called with $(H(C_1),H(C_2))$ on both $AB$ and $BA$ by $startLeader + 2\Delta$.
In the worst case, after $\max\{t_{AB},t_{BA}\}+2\Delta$, Carol should release $C_1$ to replace Alice on both $AB$ and $BA$.
This results in the replacement of the old swap hashlock $H(A_1)$ with $H(C_2)$. 
Bob does not have to call any functions to help Carol replace Alice's role. 
Thus, conforming Alice and Carol can trade without Bob's cooperation.
\end{proof}

We have proven that conforming Alice and Carol can trade the option without Bob cooperating.
Now we look at this property from another perspective.
Can Bob actively prevent Alice selling her option to Carol? 
The following result guarantees Bob cannot block Alice and Carol.

\leaderNonblockThm*

\begin{proof}
 The only way for Bob to block the transfer of an option is to contest. To contest, Bob needs to show a different mutation signed by Alice or by showing Alice's secret to the old hashlock. 
 As long as Alice is conforming to the protocol and does not collude with Bob, which is the case in this context, an adversarial Bob cannot forge her signature and successfully call \textit{contest()}.
\end{proof}

\leaderTransferAtomThm*

\begin{proof}
Conforming Carol only loses her principal if she releases her secret $C_1$. She only releases her secret $C_1$ after she sees both $AB$ and $BA$ are both mutate locked by $(H(C_1),H(C_2))$ and at least one of the conditions are met: (1)Bob has  approved both $AB$ and $BA$; (2) both of his contest windows on $AB$ and $BA$ have ended.
Bob can no longer contest this result either due to timeout of his contesting window or Bob has approved this result. 
Then when Carol sends her secret $C_1$ to both $AB$ and $BA$, she replaces Alice's role and changes the swap hashlock to $H(C_2)$ thus she owns the option. 

Conforming Alice loses her option only if Carol releases her secret $C_1$ on some contract.
We know by Lemma \ref{mutate_lock} that in the worst case \textit{mutateLockLeader()} will be called on $AB$ and $BA$ by $startLeader + 2\Delta$.
Thus, the latest time any party can call \textit{replace()} is.  $startLeader + 8\Delta$.
So if Alice loses her option there must have been a \textit{replace()} call by $startLeader + 8\Delta$.
Since $CA$ has timeout $startLeader + 10\Delta$, Alice is guaranteed enough time to claim Carol's principal.
\end{proof}

The remaining proofs are for the properties of the follower transfer protocol.

\begin{lemma}
\label{follower_bob_david_compliant}
In Protocol \ref{transfer_follower}, if Bob, David are both compliant, and Alice does nothing, then Bob gets David's principal, and David gets Bob position.
\end{lemma}
\begin{proof}
By assumption, we know the \textit{Mutate Phase} ends $startFollower+2\Delta$.
At this point, $candiate\_sender = David.address$ on $BA$ and $candidate\_receiver = David.address$ on $AB$.
By $startFollower + 3\Delta$ we know $David$ will reveal $D_1$ to $AB$ and $BA$.
At this point $sender = David$ on $BA$ and $receiver = David$ on $AB$.
Thus David has successfully taken Bob's position.
When Bob forwards $D_1$ to the $DB$, it will arrive no later than $startFollower + 4\Delta < startFollower + 5\Delta$.
So Bob gets David's principal.
\end{proof}

\followerLivenessThm*
\begin{proof}
Because Alice is compliant and doesn't reveal $A_1$, she does nothing according to Protocol \ref{transfer_follower}.
Since Bob and David are both compliant, the result follows from Lemma \ref{follower_bob_david_compliant}.
\end{proof}

\followerTransferIndThm*
\begin{proof}
Assume Alice does nothing and Bob, David are compliant.
By Lemma \ref{follower_bob_david_compliant}, the result follows immediately.
\end{proof}

\followerNonblockThm*
\begin{proof}
Alice's only ability during Protocol \ref{transfer_follower}, is to choose to reveal her secret $A_1$.
Since Bob and David are compliant, David creates $DB$ by $startFollower + \Delta$.
In the worst case, Bob calls $mutateLockFreeFollower()$ on $BA$ at the last available opportunity, at $T - 2\Delta$.
Since Bob is compliant, he also calls $mutateLockFollower()$ on $AB$ by $T - 2\Delta$.
Since David is compliant, he calls $replace()$ on $AB,BA$ by $T - \Delta$, thereby revealing $D_1$ to Bob.
Both replace calls are then completed by $startFollower + 3\Delta$.
Bob then has enough time to reveal $D_1$ on $DB$ since its timeout is $startFollower + 5\Delta$.
If Alice reveals $A_1$ by $T$ on $BA$, then David will have Bob's position so he on $AB$ and $BA$.
Thus he can reveal $A_1$ on $AB$ since it has timeout $T + \Delta$.
\end{proof}

\followerTransferAtomThm*

\begin{proof}
Suppose a conforming Bob loses his swap position.
By Protocol \ref{transfer_follower}, we know $mutateLockFollower()$ and $mutateLockFreeFollower()$ are called on $AB$ and $BA$ respectively by $startFollower + 2\Delta$ by Bob.
By assumption Bob loses his swap position.
Thus David must have revealed $D_1$ on $AB$ or $BA$ before the $replaceFollower()$ timeout.
Since each $replaceFollower()$ timeout is $2\Delta$, the latest David could have revealed $D_1$ is $startFollower + 4\Delta$.
Conforming Bob sends $D_1$ to $DB$ which arrives by $startFollower + 5\Delta$.
Since the timeout on $DB$ is $startFollower + 5\Delta$, this ensures Bob gets David's principal.
\end{proof}

\subsubsection{No UNDERWATER}
\label{underwater_proofs}

The goal of this section is to prove Theorem \ref{underwater_thm}.

If the reselling is started by Alice calling \textit{mutateLockLeader()} on $AB$ and/or $BA$ contract, the state of the corresponding contract will be changed. 
We first define states of contracts. 
These are referenced from \ref{sec:code}.
A smart contract's state is defined by the following elements:
\begin{itemize}
    \item \textit{sender}: \textit{sender} can  call \textit{refund()} if the receiver does not redeem it before swap hashlock times out.
    
    \item \textit{receiver}: \textit{receiver} can call \textit{claim()} and pass the preimage of swap hashlock to get the asset escrowed in the contract.
    
    \item  \textit{mutating}: \textit{mutating} is a bool where \textit{mutating=true} means the asset is temporarily not redeemable due to some active mutation in progress.
    In other words, the right to claim the asset is locked.
    On the other hand, \textit{mutating=false} means the asset can be claimed if the preimage to the hashlock is sent to the contract.

    \item \textit{swap\_hash\_lock}: The swap hashlock is the hashlock used to claim assets. 
    If the preimage of the hashlock is sent to $\textit{claim()}$, then the receiver can claim the asset.
    
    \item \textit{replace\_hash\_lock}: If  \textit{replace\_hash\_lock != nil}, that means Alice would like to transfer her option. \textit{replace\_hash\_lock} stores the hashlock for the transfer.
    \item \textit{candidate\_sender/candidate\_receiver}: denotes the new sender or receiver when transfer is finalized, i.e. Carol replaces Alice's role.
\end{itemize}

We take $AB$ as an example to illustrate the states. 
The states on $BA$ are symmetric which can be inferred from the context.

\begin{itemize}
    \item  $Ready2Claim(H(A_1),AB)$ denotes  state $(sender=Alice, receiver=Bob, mutating =false, swap\_hash\_lock=H(A_1))$, meaning if $A_1$ is sent before the contract expires, Bob can claim Alice's principal. 
    We will use $Ready2Claim(H(A_1))$ for simplicity when which contract we are referring to is obvious from the context.
    
    \item $MutateLockContestable(Carol,H(C_1),H(C_2),AB)$ denotes the state $(sender=Alice, receiver=Bob, mutating =true, replace\_hash\_lock = H(C_1), swap\_hash\_lock=H(C_2), candidate_receiver = Carol )$, which means Alice tentatively locks the asset and transfers her role to Carol.   
    In this state Bob is able to call \textit{contest()}.
    We will use $MutateLockContestable(H(C_1),H(C_2))$ for simplicity when other parts we are referring to are obvious from the context. 
    $MutateLockContestable$ is just a preparation for Carol to replace Alice's role. 
    The \textit{replaceLeader()} cannot be called until the state defined below.
  
    \item $MutateLockNonContestable(Carol,H(C_1),H(C_2), AB)$. 
    In this state Bob cannot call \textit{contest()} because his contest window has passed or he approved the mutation. At this state, Carol can release her secret $C_1$ to replace Alice.
    The state when Carol can replace Alice is denoted as $MutateLockNonContestable(Carol,H(C_1),H(C_2), AB)$.   
    $MutateLockNonContestable(H(C_1),H(C_2))$ is the short name.
    
    \item $Ready2Claim(H(C_2),CB)$. 
    After $MutateLockNonContestable(Carol,H(C_1),H(C_2), AB)$ is reached, Carol can call \textit{replace()} to replace Alice's role by sending $C_1$. 
    After the replacement, the state becomes $(sender=Carol, receiver=Bob, mutation\_lock=false, swap\_hash\_lock=H(C_2))$, denoted as $Ready2Claim(H(C_2),CB)$.
    $Ready2Claim(H(C_2))$ is the short name.
    
    \item $Claimed(H(S),AB)$ means the asset on $AB$ contract is claimed by a secret $S$ which is the preimage to $H(S)$ and Bob gets the asset. 
    In this context, $S$ can be $A_1$ or $C_2$.
    $Claimed(H(S))$ is the short name.
\end{itemize}

The state transition is depicted in Fig~\ref{fig:state_transition}.

\begin{figure}
    \centering
    \includegraphics[width = \linewidth]{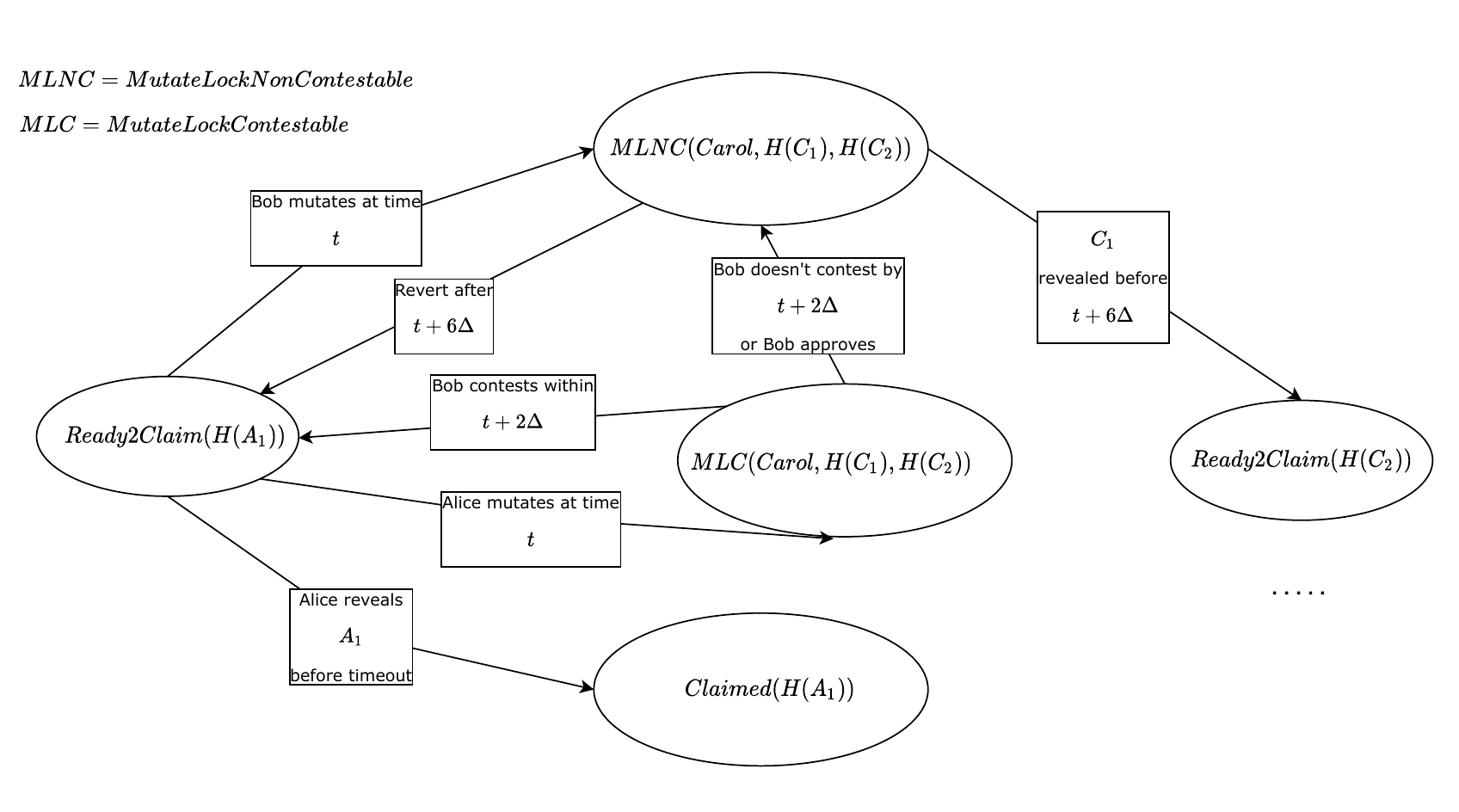}
    \caption{State transition until Carol replaces Alice on $AB$ contract }
    \label{fig:state_transition}
\end{figure}

There are two stages of changes of state on the $AB$ and $BA$ contracts. 

\begin{enumerate}
    \item First stage. 
    Alice and Bob aim to change the state of both contracts from $Ready2Claim(H(A_1))$ to a mutation state $MutateLockNonContestable(H(C_1),H(C_2))$, during which  state $MutateLockContestable(H(C_1),H(C_2))$ may be reached temporally, corresponding to Mutate Lock Phase and Consistency Phase.
     
    \item Second stage.  
    Starting from $MutateLockNonContestable(H(C_1),H(C_2))$, Bob and Carol are involved in an atomic change from  $MutateLockNonContestable(H(C_1),H(C_2),Carol)$ to $Ready2Claim(H(C_2))$ by Carol releasing a secret $C_1$. 
    It is similar to a standard atomic swap where Carol is the leader and Bob is the follower, corresponding to the Replace/Revert Phase.
\end{enumerate}

\begin{theorem}
Conforming Alice will never end up in \textit{UNDERWATER}.
\end{theorem}
\begin{proof}
If Carol never creates $CA$ or creates $CA$ with different conditions than what Carol and Alice agreed upon, then Alice does nothing and so she doesn't lose her option.

Otherwise, consider the case when Carol creates $CA$ with the conditions Alice and Carol had agreed upon.

If  Alice is conforming, by Lemma \ref{mutate_lock} then both $AB$ and $BA$ are mutate locked with $(H(C_1),H(C_2))$ by $startLeader + 2\Delta$.
This means $AB$ and $BA$ both have reached the state $MutateLockContestable(H(C_1),H(C_2))$ by $startLeader + 2\Delta$.
Since Alice is conforming, by $startLeader + 4\Delta$, the state will transition to $MutateLockNonContestable(H(C_1),H(C_2))$ since Bob cannot forge Alice's signature and contest.

After $MutateLockNonContestable(H(C_1),H(C_2))$, if Carol releases $C_1$, Alice observes it in the Replace/Revert phase, and can get Carol's principal as shown in Theorem \ref{leader_transfer_atom_thm}.
If Carol does not release $C_1$, then eventually, after $t_{AB}+6\Delta$ and $t_{BA}+6\Delta$ respectively, the states on both contracts can be reverted to $Ready2Claim(H(A_1))$ which means Alice still owns the option and it is unlocked. If Alice finally owns the option, then she does not end up in \textit{UNDERWATER} by the guarantee of atomic swap. If Alice loses her option(maybe only lose a role on one contract), she gets Carol's principal, with/without losing her principal escrowed in her old option, which is acceptable for her.
\end{proof}

\begin{theorem}
Conforming Carol will never end up in \textit{UNDERWATER}.
\end{theorem}
\begin{proof}
After Carol escrows her principal to Alice on $CA$ and sends her hashlocks $H(C_1),H(C_2)$ to Alice, Carol does not do anything in the first stage but observe.
She only joins in the protocol after she observes that both contracts are in  $MutateLockNonContestable(H(C_1),H(C_2))$ state. 
Otherwise, she is silent and she will not end up in \textit{UNDERWATER} since her principal will be refunded eventually. 
After she observes $MutateLockNonContestable(H(C_1),H(C_2))$ state on both contracts, she releases $C_1$, moving $MutateLockNonContestable(H(C_1),H(C_2))$ state to $Ready2Claim(H(C_2))$. Then she owns the option provided by Bob and takes over Alice's option.
In that case, she may lose her principal to Alice due to the release of $C_1$.
Say, in the worst, she loses her principal to Alice.

Then, if Carol decides to exercise the option, she gets Bob's principal. 
That means, her principal is exchanged with Bob's. 
If she does not exercise the option and let it expire, she can get Alice's principal. That means, her principal is exchanged with Alice's. In any case, Carol never ends up in \textit{UNDERWATER}.
\end{proof}

To prove that conforming Bob does not end up in \textit{UNDERWATER}, we need to to prove the consistency of states on two contracts. Specifically,

\begin{itemize}
    \item In the first stage, if Alice calls \textit{mutateLockLeader()} in any contract, either both contracts eventually gets $MutateLockNonContestable(H(C_1),H(C_2)$ or reverted back to $Ready2Claim(H(A_1))$. 
    \item  In the second stage, if Carol releases $C_1$, both contracts are eventually at state $Ready2Claim(H(C_2))$. 
    If Carol does not releases $C_1$,  both contracts are eventually are in state $Ready2Claim(H(A_1))$.
\end{itemize} 

Suppose, without loss of generality, that $AB$ contract is the first to update its state from $Ready2Claim(H(A_1))$ to a new state by Alice calling \textit{mutateLockLeader()} at some time $t$, assuming Bob is conforming.
\begin{theorem}
If $BA$ is not claimed by $t+\Delta$, then $BA$ will be mutate locked by $t+\Delta$. 
\end{theorem}
\begin{proof}
Since Bob is conforming, when Alice calls \textit{mutateLockLeader()} on the $AB$ contract, Bob can send \textit{mutateLockLeader()} transaction on the $BA$ within $\Delta$. Recall that the start time that \textit{mutateLockLeader()} is called is denoted as $t_{AB}$ and $t_{BA}$ respectively.
Then in $BA$ \textit{mutateLockLeader()} is called on $BA$ by $t_{AB}+\Delta$.

If Bob's \textit{mutateLockLeader()} is included on $BA$, then $t_{BA}\leq t_{AB}+\Delta$.
If Bob's \textit{mutateLockLeader()} on $BA$ is not included, then it must be Alice was able to call it before he was. 
In either case, $t_{BA}\leq t_{AB}+\Delta$. 
\end{proof}

The case that by $t+\Delta$, the other contract's state changes to $Claimed(H(A_1))$ will be analyzed in Theorem \ref{theorem:goodbob}.

\begin{theorem}\label{theorem:after2.0}
Suppose two contracts both eventually reach  $MutateLockNonContestable(H(C_1),H(C_2))$. Then, either both contracts are eventually 
$Ready2Claim(H(A_1))$ or $Ready2Claim(H(C_2))$.
\end{theorem}
\begin{proof}
  Since we know the start time of mutation between two contracts does not stagger beyond $\Delta$. Denote the start time on two contracts as $t_{first}$ and $t_{second}$ respectively where $t_{first}\leq t_{second}$. Since Bob is conforming, $ t_{second}-t_{first}\leq \Delta$. We denote the timeout for Carol to release $C_1$ on $AB$ and $BA$ contracts as $t_1$,$t_2$ respectively, and the timeout for Bob to release $C_1$ as $t_3$,$t_4$. Without loss of generality, assume $t_1=t_{first}+4\Delta$ and $t_2=t_{second}+4\Delta$. $t_3=t_{first}+6\Delta$ and $t_4=t_{second}+6\Delta$.
  We see that the latest timeout for Carol to release $C_1$ is $t_2$.  
  The earliest timeout for Bob to send $C_1$ satisfies $t_3\geq t_2+\Delta$ since we have $t_3=t_{first}+6\Delta\geq t_{second}+5\Delta=t_2+\Delta$. 
  That means,  if Carol releases $C_1$, then it is eventually sent to both contracts and the state is $Ready2Claim(H(C_2))$. 
  If Carol does not release $C_1$, then both contracts are reverted to $Ready2Claim(H(A_1))$ after timeout $t_3$ and $t_4$, respectively.
\end{proof}

\sloppy The atomic changes between AB and BA are a bit more complicated, since starting from $Ready2Claim(H(A_1))$, there are multiple possible state changes available on both contracts: $Claimed(H(A_1))$, $MutateLockContestable(H(C_1),H(C_2))$, and  $MutateLockNonContestable(H(C_1),H(C_2))$. The states $Claimed(H(A_1))$ and $MutateLockNonContestable(H(C_1),H(C_2))$ are not contestable. Assuming Bob is conforming and relay whatever he sees from one contract to another contract.

\begin{lemma}\label{lemma:2delta}
If a contract is transferred to $MutateLockContestable(H(C_1),H(C_2))$ at $t$, then within $t+2\Delta$, it learns the other contract's update and then agree on either reverting back if there is a conflict or agree on the same mutation if no conflict.
\end{lemma}
\begin{proof}
 If a contract, say AB, is transferred to
$MutateLockContestable(H(C_1),H(C_2))$ at time $t$, since the compliant Bob relays it, whatever happens before $t-\Delta$ to AB contract, that means, at time $t-\Delta$, the other chain is in $Ready2Claim(H(A_1))$. If there is any change happening on BA contract, it would happen between $t-\Delta$ to $t+\Delta$. Then, by $t+2\Delta$,  AB contract will learn what happens on BA contract between $t-\Delta$ to $t+\Delta$ by Bob. If AB contract does not receive any transaction from Bob, then no conflicting changes happened on BA from  $t-\Delta$ to $t+\Delta$. Bob can update BA contract to $MutateLockContestable(H(C_1),H(C_2))$ by relaying Alice's mutation to BA contract,  and  AB contract can be updated to $MutateLockNonContestable(H(C_1),H(C_2))$ after $t+2\Delta$. If there is any state change between  $t-\Delta$ to $t+\Delta$ on BA contract, then by $t+2\Delta$, AB contract would be informed, and its own state update is also sent to BA contract. Then the conflicting state update would revert back their state change to $Ready2Claim(H(A_1))$ when it receives the conflicting change.
\end{proof}

\begin{theorem}\label{theorem:goodbob}
Conforming Bob never ends up in \textit{UNDERWATER}.
\end{theorem}
\begin{proof}
  Starting from $Ready2Claim(H(A_1))$, 
\begin{enumerate}
    \item If a contract is transferred to $MutateLockContestable(H(C_1),H(C_2))$ at $t$, and the other contract is not updated to $Claimed(H(A_1))$, then by Lemma \ref{lemma:2delta}, we know that either both contracts eventually reach \\ $MutateLockNonContestable(H(C_1),H(C_2))$ or are reverted back to $Ready2Claim(H(A_1))$.
    
    \item If a contract BA is transferred to $Claimed(H(A_1))$ state at $t$, then, if the other contract makes temporary $MutateLockContestable(H(C_1),H(C_2)$, it will be revert back to $Ready2Claim(H(A_1))$ by $t+\Delta$, and then transfer to  $Claimed(H(A_1))$. Bob does not end up with \textit{UNDERWATER}.

    \item If a contract is transferred directly to $MutateLockNonContestable(H(C_1),H(C_2))$ at $t$, then the other contract must be $MutateLockContestable(H(C_1),H(C_2))$ at $t'\in [t-\Delta, t]$, and then after $t+2\Delta$ or sooner, the other contract will become $MutateLockNonContestable(H(C_1),H(C_2))$.
\end{enumerate}
When $MutateLockNonContestable(H(C_1),H(C_2))$ is reached on both contracts, by Theorem \ref{theorem:after2.0}, Bob will either be involved in the swap with Alice by Alice's secret $A_1$ or involved in the swap with Carol by Carol's secret $C_2$. By the guarantee of atomic swap, Bob will not end up with \textit{UNDERWATER}.

\end{proof}

The details for showing Protocol \ref{transfer_follower} satisfies \textit{No UNDERWATER} are omitted since they mirror the previous arguments shown for \ref{transfer_leader}.